\newif\iflncs
\newif\ifnotes

\newif\ifblind\blindfalse

\notesfalse

\lncsfalse

\iflncs
\RequirePackage{amsmath}
\documentclass{llncs} 
\usepackage{StyleSheets/template_lncs}
\usepackage{StyleSheets/global}
\else
\documentclass[runningheads, 11pt]{article}
\usepackage{times}
\usepackage{physics}
\usepackage{StyleSheets/template}
\usepackage{StyleSheets/global}
\fi

\usepackage{global}
\usepackage{graphicx} 
\usepackage{amsfonts}
\usepackage{amsthm}
\usepackage{amsmath}
\usepackage[operators,sets]{cryptocode}
\usepackage[margin=1in]{geometry}

\theoremstyle{plain}
\newtheorem{problemnormal}{Problem}

\title{SVP$_p$ is Deterministically NP-Hard for all $p > 2$, \\
Even to Approximate Within a Factor of $2^{\log^{1-\varepsilon} n}$}
\author{Isaac M Hair \thanks{\texttt{isaacmhair@gmail.com}} \\ UCSB, UCLA \and Amit Sahai \thanks{\texttt{sahai@cs.ucla.edu}} \\ UCLA}
\date{}

\begin{document}

\maketitle
\noteswarning

\thispagestyle{empty}

\begin{abstract}
    We prove that SVP$_p$ is NP-hard to approximate within a factor of $2^{\log^{1 - \varepsilon} n}$, for all constants $\varepsilon > 0$ and $p > 2$, under standard deterministic Karp reductions. This result is also the first proof that \emph{exact} SVP$_p$ is NP-hard in a finite $\ell_p$ norm. Hardness for SVP$_p$ with $p$ finite was previously only known if NP~$\not \subseteq$~RP, and under that assumption, hardness of approximation was only known for all constant factors. As a corollary to our main theorem, we show that under the Sliding Scale Conjecture, SVP$_p$ is NP-hard to approximate within a small polynomial factor, for all constants $p > 2$.
    
    Our proof techniques are surprisingly elementary; we reduce from a \emph{regularized} PCP instance directly to the shortest vector problem by using simple gadgets related to Vandermonde matrices and Hadamard matrices.
\end{abstract}

\newpage

\setcounter{page}{1}

\section{Introduction}

A lattice $\mathcal{L}$ is the set of all integral linear combinations of a set of basis vectors $\{\mat b_1, \ldots, \mat b_n\}$. Lattices are a classical mathematical object, with research on the subject dating back centuries. In the last several decades, lattices have come to the forefront of study in theoretical computer science, due to their applications in algorithmic number theory \cite{lenstra1982factoring}, integer programming \cite{lenstra1983integer, kannan1987minkowski, frank1987application, schrijver1998theory}, coding theory \cite{forney1988coset, de2002some, zamir2014lattice}, and perhaps most strikingly, post-quantum cryptography \cite{ajtai1998shortest, nguyen2001two, micciancio2007worst, regev2009lattices, peikert2016decade}. Most cryptographic primitives based on lattice problems enjoy a special feature: while the problems used in actual cryptosystems are sampled from random distributions, the computational hardness of breaking these cryptosystems is implied by the difficulty of solving \emph{worst-case} problems. Perhaps the most foundational of these worst-case problems is the gap shortest vector (GapSVP) problem. As the name suggests, this problem amounts to approximating the length of the shortest nonzero vector in a given lattice $\mathcal{L}$, where we measure the length of a vector $\mat w \in \mathcal{L}$ using the $\ell_p$ norm:\footnote{
Recall that    for finite $p$, we define the $\ell_p$ norm of a vector $\mat w \in \mathbb{Z}^m$ as
$\|\mat w\|_p \coloneqq \left(\sum_{i = 1}^m \vert \mat w_i\vert^p\right)^{1/p},$
    and for $\ell_\infty$ we define
    $\|\mat w\|_\infty \coloneqq \max_i \vert\mat w_i\vert.$
    We use the $\ell_0$ \emph{pseudo-norm} to denote the number of nonzero entries in a given vector.
}

\begin{problemnormal}[$\gamma$-GapSVP$_p$]
    Given a matrix $\mat M \in \mathbb{Z}^{n \times m}$ and a threshold $k$, distinguish between the following two cases, when one of them is promised to hold:
    \begin{enumerate}
        \item There exists a nonzero vector $\mat v \in \mathbb{Z}^n$ such that $\|\mat v\mat M\|_p \leq k$.
        \item For all nonzero vectors $\mat v \in \mathbb{Z}^n$, we have $\|\mat v \mat M\|_p > \gamma k$.
    \end{enumerate}
    (Here and throughout the paper, all vectors are row vectors.)
\end{problemnormal}

There is a rich body of work exploring upper bounds \cite{kannan1987minkowski, ajtai2002sampling, blomer2009sampling, eisenbrand2022approximate} and lower bounds \cite{micciancio2001shortest, haviv2007tensor, aggarwal2018gap, bhattiprolu2024inapproximability} for all choices of $p$, as well as the relationship between GapSVP instances in different $\ell_p$ norms \cite{regev2006lattice, aggarwal2021dimension}. Besides serving as a natural counterpart to the results for $p = 2$, GapSVP results for $p \neq 2$ have proved important historically because they often lead to results for the $p = 2$ case. This is perhaps most obvious when examining the chronology of GapSVP lower bounds.

The shortest vector problem (without an approximation gap) was first proven to be NP hard in the $\ell_\infty$ norm by van Emde Boas \cite{van1981another}, who conjectured that the same hardness result should hold for $\ell_2$. After nearly two decades, Ajtai \cite{ajtai1998shortest} gave a breakthrough reduction showing that the exact shortest vector problem is hard in the $\ell_2$ norm
unless \textnormal{NP $\subseteq$ RP}, i.e. with the caveat that his reduction is \emph{randomized}. Over the next few years, researchers improved upon Ajtai's results \cite{cai1998approximating, micciancio2001shortest, khot2003hardness, khot2005hardness, haviv2007tensor}. Many of the improvements were first discovered for $\ell_p$ with $p$ being a sufficiently large constant, and then were generalized to all $p \geq 1$. The current state of affairs with regard to the assumption that $\textnormal{NP} \not \subseteq \textnormal{RP}$ is summarized in the following theorem.

\begin{theorem} [\cite{khot2005hardness}]
    For all $p \geq 1$, there is no randomized polynomial time algorithm for $\gamma$-GapSVP$_p$, for $\gamma = O(1)$, unless \textnormal{NP $\subseteq$ RP}.
\end{theorem}
Hardness of approximation for larger (but still subpolynomial) factors are known only under significantly complexity stronger assumptions like \textnormal{NP $\not \subseteq$ RTIME($2^{(\log n)^{O(1)}}$)} or  \textnormal{NP $\not \subseteq$ $\bigcap_{\delta > 0}$RTIME($2^{n^\delta}$)} ~\cite{khot2005hardness,regev2006lattice, haviv2007tensor}.
All known reductions are randomized, and the best approximation factor we can achieve under NP $\subseteq$ RP is still just $\gamma = O(1)$. There is a single exception: Dinur showed that $\gamma$-GapSVP$_\infty$ is NP hard under a deterministic reduction, where $\gamma = n^{c/\log\log n}$ for a sufficiently small constant $c > 0$ \cite{dinur2002approximating}.

Thus, a central 40-year-old open question is whether we can get a deterministic NP-hardness reduction when $p$ is finite, which was explicitly asked by many authors including van Emde Boas \cite{van1981another}, Micciancio \cite{micciancio2001shortest, micciancio2012inapproximability}, Haviv and Regev \cite{haviv2007tensor}, and Bennett and Peikert \cite{bennett2022hardness, bennett2023complexity}.

\begin{question}[\textnormal{{(}\cite{bennett2022hardness}{)}}] \label{q:1}
    ``Indeed, it is a notorious, long-standing open problem to prove that GapSVP$_p$ is NP-hard, \textbf{even in its exact form},\footnote{Emphasis added.} under a deterministic reduction for some finite $p$.''
\end{question}

Another natural question is whether we can improve the NP hardness of approximation factor, even under randomized reductions.

\begin{question}[\textnormal{\emph{(Implicit in the complexity community)}}] \label{q:2}
    Can we find a polynomial time (potentially randomized) reduction from an NP complete problem to $\omega(1)$-GapSVP$_p$, for some finite $p$?
\end{question}

One explanation for the lack of improvement since Haviv and Regev's work is that essentially all known hardness of approximation results for GapSVP$_p$ (with the exception of $p = \infty$) boil down to a similar technique: They first achieve a constant gap via a (randomized) polynomial time reduction from the closest vector problem, which is already known to be NP-hard to approximate~\cite{arora1997hardness}, by using a special mathematical object called a \emph{locally dense lattice}. Then they apply tensoring to boost the approximation factor past a constant. Due to the critical usage of tensoring, this approach appears to be incapable of yielding hardness results for $\gamma = \omega(1)$ under the sole assumption that NP $\subseteq$ RP. (Each tensoring operation increases the lattice dimension by a multiplicative factor of poly$(n)$, and we need a superconstant number of such operations to reach $\gamma = \omega(1)$.) Derandomization also appears to be quite difficult; despite decades of research, we do not have a deterministic construction of locally dense lattices.

\paragraph{A New More Direct Approach.}
In this paper, we resolve both questions above in the affirmative. We present a \emph{deterministic, nearly-polynomial factor NP hardness of approximation result for $\gamma$-GapSVP$_p$ for all constants $p > 2$}, based on a direct reduction from a carefully pre-processed probabilistically checkable proof (PCP) for NP. Our reduction bypasses the closest vector problem entirely; its closest relative appears to be Dinur's 2002 NP hardness of approximation result for $\gamma$-GapSVP$_\infty$. Before presenting our results more formally, we give some background on PCPs.

A PCP consists of a \emph{verifier} which takes as input a sequence of symbols over alphabet $\Sigma$ (the \emph{proof}) and behaves as follows. The verifier uses $r$ random bits to select $q$ locations in the proof to read, and accepts iff the symbols at those locations satisfy some constraint (which may or may not be different for each choice of random bits). Finally, a (perfect correctness) PCP theorem provides a deterministic polynomial time reduction from instances of the NP complete problem circuit SAT to the \emph{description} of a polynomial-time verifier, so that:
\begin{enumerate}
    \item (Completeness) If the SAT instance was satisfiable, then there exists a proof which causes the verifier to accept with probability 1.
    \item (Soundness) If the SAT instance was unsatisfiable, then every proof will cause the verifier to accept with probability at most $s$ (the \emph{soundness parameter}).
\end{enumerate}

The strongest constant-query PCP we currently know how to construct is due to Dinur, Fischer, Kindler, Raz, and Safra:

\begin{theorem}[\cite{dinur1999pcp}] \label{thm:PCPbeta}
    For every constant $\varepsilon > 0$, there exists a PCP for SAT instances of size $n$ that uses $r = O(\log n)$ random bits to make $q = O(1)$ queries to a proof over alphabet $\Sigma$, where the alphabet size is $\vert\Sigma\vert = 2^{\Theta(\log^{1 - \varepsilon} n)}$ and the soundness parameter is $s = 2^{-\Theta(\log^{1 - \varepsilon} n)}$. The PCP can be constructed deterministically in $n^{O(1)}$ time.
\end{theorem}

For our result, we will not leverage the PCP theorem exactly as stated above. Instead we first pre-process the PCP using a regularization technique due to Hirahara and Moshkovitz \cite{hirahara2023regularization}, which ensures that every position in the proof is read with equal probability. Details on this technique are given in Section \ref{sec:regPCP}.

At a very high level, our reduction will make use of simple, deterministically constructable gadgets related to Vandermonde matrices and Hadamard matrices.\footnote{Our Hadamard matrix gadgets bear topical resemblance to the ``norm measuring rows'' in Dinur's NP hardness of approximation result for $\gamma$-GapSVP$_\infty$ \cite{dinur2002approximating}, but it appears that we leverage properties of the gadgets not discussed in the SVP literature before. Additionally, Dinur's hardness result for $\gamma$-GapSVP$_\infty$ relies on a novel characterization of NP in terms of a gap problem referred to as SSAT, which is very different from the regularized PCP we leverage.} Our analysis involves a novel approach for examining how these gadgets interact specifically within the context of a regularized PCP. The construction is not very technical; rather, it requires a paradigm shift in how we think about reductions to lattice problems.

\subsection{Our Results}

Below we give a basic version of our main theorem; for the full statement, see Theorem \ref{thm:mainresultformal}.

\begin{theorem}[Informal] \label{thm:mainresultinformal}
    Suppose there exists a PCP for SAT instances of size $n$ that uses $r = O(\log n)$ random bits to make $q = O(1)$ queries to a proof over alphabet $\Sigma$ and has soundness $s$, where $1/s \geq r^{\omega(1)}$. Suppose further that every proof location is read with the same probability, and that the PCP can be constructed deterministically in polynomial time. Then for all constants $p > 2$ (and for $p = \infty$), $\gamma$-GapSVP$_p$ on $n^{O(1)}$ dimensional lattices is NP hard, where $\gamma = (1/s)^{\Theta(1)}$.
\end{theorem}

In Section \ref{sec:puttingittogether}, we show that the PCP from Theorem \ref{thm:PCPbeta} can be regularized via the techniques of \cite{hirahara2023regularization} to get a PCP that satisfies all preconditions of Theorem \ref{thm:mainresultinformal}, where the alphabet size is $\vert\Sigma\vert = 2^{\Theta(\log^{1 - \varepsilon}n)}$ and the soundness is $s = 2^{-\Theta(\log^{1-\varepsilon}n)}$. From this along with Theorem \ref{thm:mainresultinformal} we derive the following.

\begin{corollary} \label{cor:NPhard}
    For all constants $\varepsilon > 0$ and $p > 2$, $\gamma$-GapSVP$_p$ on $n$ dimensional lattices is NP hard, where $\gamma = 2^{\log^{1 - \varepsilon} n}$.
\end{corollary}

\begin{remark}
    We suspect that, by adapting the PCP given by Dinur, Harsha, and Kindler \cite{dinur2015polynomially}, the same NP hardness result should hold for slightly \emph{subconstant} values of $\varepsilon$. We do not pursue such an improvement in this paper.
\end{remark}

Corollary \ref{cor:NPhard} succinctly resolves both Question \ref{q:1} and Question \ref{q:2} in the affirmative. Due to the elementary nature of our techniques (see Section \ref{sec:setup} and Section \ref{sec:actualconstruct}), we believe that our approach should inspire hardness of approximation results for several other lattice problems, hopefully paving the way\footnote{It's also worth noting that we found our hardness of approximation result as a byproduct of research on new sources of hardness to build cryptographic primitives such as public key encryption and witness encryption; we expect the connection between complexity-theoretic problems and cryptographic problems to inspire even more results.} to a deterministic NP hardness of approximation result for $\gamma$-GapSVP$_2$. 

Another consequence of Theorem \ref{thm:mainresultinformal} is that we can get polynomial factor NP hardness of approximation under the Sliding Scale Conjecture~\cite{bellare1993efficient}, which posits that there exists a PCP for SAT with the same parameters as in Theorem \ref{thm:PCPbeta}, except $\vert \Sigma\vert = n^{\Theta(1)}$ and $s = n^{-\Theta(1)}$.\footnote{It's known that any non-degenerate PCP requires $\vert \Sigma\vert \geq (1/s)^{\Omega(1/q)}$, so when q is constant we cannot hope to achieve polynomially small soundness without a polynomially large alphabet.} Details are given in Section \ref{sec:puttingittogether}.

\begin{corollary} \label{cor:sliding}
    Assuming the Sliding Scale Conjecture, $\gamma$-GapSVP$_p$ on $n$ dimensional lattices is NP hard for all constants $p > 2$ (and for $p = \infty$), where $\gamma = n^c$ for a constant $c > 0$ that depends on $p$.
\end{corollary}

Mukhopadhyay \cite{mukhopadhyay2022projection} recently proved a weaker version of this result: she showed that the \emph{Projection Games Conjecture} implies NP hardness of $n^{\Omega(1)}$-GapSVP$_\infty$, that is, for the special case that $p = \infty$. This alternative conjecture posits that the Sliding Scale Conjecture holds even for projection games, which are a special type of PCP in which every constraint depends on $q = 2$ variables, and additionally every constraint is of a particular form \cite{moshkovitz2012projection}. To get a sense of the relative strength of these two conjectures, observe that the Projection Games Conjecture implies the Sliding Scale Conjecture, but it is unknown whether the reverse implication holds. Indeed, the current best polynomial time reduction from SAT to a projection game achieves soundness $(\log n)^{-O(1)}$ \cite{dinur2014analytical}, as opposed to the $2^{-(\log n)^{0.99}}$ soundness for general PCPs.

\section{Preliminaries}
\label{sec:regPCP}

We use $[n]$ to denote the set $\{1, \ldots, n\}$. All vectors are row vectors unless stated otherwise. Assume that all quantities are rounded integers as needed.

\paragraph{Regularized PCPs.} Regularization is a classic tool for getting hardness of approximation results from the PCP theorem. The idea is to convert an arbitrary PCP into a new PCP where every variable in the proof is read with equal probability. More than three decades ago, Papadimitriou and Yannakakis showed that any PCP can be converted into a regularized PCP, but with the caveat that the final PCP has at best constant soundness \cite{papadimitriou1988optimization}. Recently, Hirahara and Moshkovitz gave an elegant technique which allows the final PCP to have soundness that is polynomially related to that of the starting PCP \cite{hirahara2023regularization}. For the convenience of the reader, we prove the following result in Appendix \ref{app:regular}.

\begin{theorem}[Simplified version of theorem in \cite{hirahara2023regularization}] \label{thm:regtechnique}
    Assume that we have a PCP verifier $\mathcal{V}$ which uses $r$ random bits to make $q$ queries to a proof over alphabet $\Sigma$ and has soundness $s$. Assume further that $s \leq \min(1/(3q), 1/\vert\Sigma\vert^c)$, where $c$ is a constant with $0 < c \leq 1$. Then we can construct a new PCP verifier $\mathcal{V}'$ from $\mathcal{V}$ deterministically in $\vert \mathcal{V}\vert^{O(1)}$ time, such that the following holds. $\mathcal{V}'$ uses $r + O(\log(1/s))$ random bits to make $q^{O(1)}$ queries to a proof over the same alphabet $\Sigma$ and has soundness $s^{\Theta(1)}$. Additionally, every proof symbol is read for exactly $d = (q/s)^{\Theta(1)}$ choices of randomness.
\end{theorem}

\paragraph{Hölder's Inequality.}
This inequality generalizes the Cauchy-Schwarz inequality.

\begin{theorem}[\cite{holder1889ueber}] \label{thm:Holder}
    Let $\mat u, \mat v \in \mathbb{Z}^n$, and let $p, q \geq 1$ satisfy $1/p + 1/q = 1$. Then
    \[\|\mat u \odot \mat v\|_1 \leq \|\mat u\|_p \|\mat v \|_q,\]
    where $\odot$ denotes the component-wise product.
\end{theorem}

We make extensive use of the following corollary, the proof of which is deferred to Appendix \ref{app:Holder}.

\begin{corollary} \label{cor:Holder}
    For all $\mat w \in \mathbb{Z}^n$ and $p \geq 2$ (including $p = \infty$),
    \[\|\mat w\|_p \geq n^{1/p - 1/2}\|\mat w\|_2.\]
\end{corollary}

\section{Setting up the Gadgets} \label{sec:setup}

In this section, we give an overview of our tools and briefly discuss how we plan to use them to reduce from a regularized PCP to the shortest vector problem.

\subsection{Viewing PCPs as CSPs} \label{sec:cspperspective}

A helpful way to think of PCPs, and indeed the viewpoint we will adopt for our reduction, is to interpret them as showing that a certain constraint satisfaction problem (CSP) is hard to approximate. We think of every choice of randomness as corresponding to a distinct constraint in the CSP, and every position in the proof as a variable in the CSP. There are exactly $M = 2^r$ choices of randomness, so our CSP will have exactly $M$ constraints. For every choice of randomness, there is an explicit list of assignments for the relevant proof symbols that would cause the verifier to accept; these become the satisfying assignments for the corresponding constraint in the CSP. To emphasize the CSP viewpoint, we refer to a PCP with $M = 2^r$ choices of randomness that reads a proof of length $N$ as ``a PCP having $M$ constraints and $N$ variables.'' Assuming the PCP is regular, every variable will appear in exactly $d$ constraints and every constraint will depend on exactly $q$ variables, so we have $Nd = Mq$.

If the SAT instance from which we constructed the PCP was satisfiable, then by definition there must exist a proof which makes the verifier accept with probability 1. From our CSP viewpoint, this corresponds to an assignment to the variables such that every constraint is satisfied. On the other hand, if the original SAT instance was not satisfiable, then for any purported proof the PCP verifier will accept with probability at most the soundness parameter $s$. From the CSP perspective, this means that any assignment to the variables will satisfy at most an $s$ fraction of the constraints.

\paragraph{Associated Graphs and Matrices.} Below, we define a graph which corresponds to the CSP. The graph will encode which constraints are incident to which variables, as well as which variable assignments satisfy each constraint.

\begin{definition}[Label Extended Factor Graph]
    Consider a PCP with $N$ variables, $M$ constraints, and alphabet $\Sigma$, where each constraint depends on $q$ variables. Its label extended factor graph is the bipartite graph $([M] \times \Sigma^q, [N] \times \Sigma, E)$, where $((t, \alpha_1, \ldots \alpha_q), (x, \alpha)) \in E$ iff
    \begin{enumerate}
        \item The (ordered) assignment $\alpha_1, \ldots \alpha_q$ is a satisfying assignment for constraint $t$,
        \item Constraint $t$ contains the variable $x$, and
        \item The (ordered) assignment $\alpha_1, \ldots \alpha_q$ indicates an assignment of $\alpha$ to the variable $x$.
    \end{enumerate}
\end{definition}

In other words, each left vertex $(t, \alpha_1, \ldots, \alpha_q)$ represents a \emph{candidate} assignment for constraint $t$, and we list every such candidate. Similarly, each right vertex $(x, \alpha)$ represents a (variable, assignment) tuple. The left vertices $(t, \alpha_1, \ldots, \alpha_q)$ which do not correspond to a satisfying assignment for constraint $t$ will have no edges incident to them. (The reason we keep these vertices is simply to ensure the label extended factor graph has exactly $M\vert\Sigma\vert^q$ left vertices and $N\vert\Sigma\vert$ right vertices.)

We can also think of the label extended factor graph in terms of its indicator matrix:

\begin{definition}[Indicator Matrix]
    The indicator matrix $\mat S \in \{0, 1\}^{M\vert\Sigma\vert^q \times N\vert\Sigma\vert}$ for a given PCP has \\ $\mat S_{(t, \alpha_1, \ldots, \alpha_q), (x, \alpha)} = 1$ iff $((t, \alpha_1, \ldots, \alpha_q), (x, \alpha))$ is an edge in the PCP's label extended factor graph.
\end{definition}

Below, we give some toy examples of indicator matrices.\\

\noindent\begin{minipage}{0.23\textwidth}
\includegraphics[width=\linewidth]{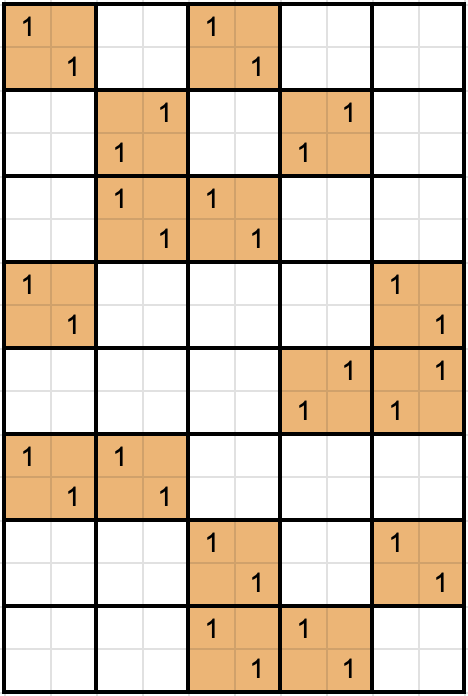}
\end{minipage}%
\hfill%
\begin{minipage}{0.02\textwidth}
\end{minipage}%
\hfill%
\begin{minipage}{0.28\textwidth}
The indicator matrix for a PCP that has $M = 8$ constraints, $N = 5$ variables, $q = 2$ variables read by each constraint, and $\vert\Sigma\vert = 2$. For ease of viewing, we omit rows that are entirely zero, and we only write the ``1'' entries. Notice that this PCP is satisfied by the assignment that corresponds to picking the left column for each variable.
\end{minipage}%
\hfill%
\begin{minipage}{0.02\textwidth}
\end{minipage}%
\hfill%
\begin{minipage}{0.23\textwidth}
\includegraphics[width=\linewidth]{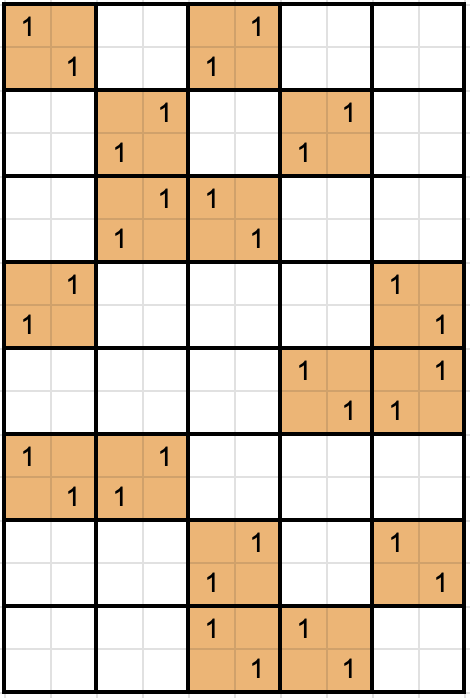}
\end{minipage}%
\hfill%
\begin{minipage}{0.02\textwidth}
\end{minipage}%
\hfill%
\begin{minipage}{0.18\textwidth}
The indicator matrix for a PCP with the same parameters as the one on the left, but which is unsatisfiable.
\end{minipage}%
\hfill%
\begin{minipage}{0.02\textwidth}
\end{minipage}%
\hfill

\subsection{Matrix Gadgets}

We form the set of basis vectors for our SVP problem by taking the indicator matrix $\mat S$ for a regularized PCP, and then augmenting $\mat S$ using different types of gadgets. (A modified copy of the matrix $\mat S$ will directly appear within the final set of basis vectors.) In this section we introduce the gadgets themselves, with a particular focus on how they can help us ensure soundness of the reduction.

\paragraph{Reduced Vandermonde Matrices.}
The first gadget is a matrix where all entries are polynomially bounded integers, and every row-induced square submatrix has full rank.

\begin{definition}[Reduced Vandermonde Matrix]
    Let $a > b$ with $a$ prime. We define an $(a, b)$ reduced Vandermonde matrix $\mat V \in \mathbb{Z}^{(a - 1) \times b}$ as $\mat V_{i, j} = i^{j-1} \mod a$.
\end{definition}

We have the following property:

\begin{lemma}
    Every $b \times b$ submatrix of an $(a, b)$ reduced Vandermonde matrix is of full rank.
\end{lemma}

\begin{proof}
    Well-known; see for example \cite{horn1994topics}. The idea is that an $(a - 1) \times b$ Vandermonde matrix over $\mathbb{F}_a$ has every $b \times b$ submatrix being of full rank, and casting from $\mathbb{F}_a$ to $\mathbb{Z}$ cannot introduce new linear dependencies.
\end{proof}

Thus, any linear combination of the rows which sums to zero must either (i) have all coefficients equal to zero, or (ii) have a large fraction of the coefficients being nonzero. More formally,

\begin{corollary} \label{cor:Vandermonde}
    Let $\mat v \in \mathbb{Z}^{a-1}$ be a nonzero vector. If $\mat v \mat V = \mat 0$, then $\|\mat v\|_0 > b$.
\end{corollary}

Reduced Vandermonde matrices are used in two different parts of our construction:
\begin{enumerate}
    \item We concatenate a single large reduced Vandermonde matrix (with its entries scaled up by a large multiplicative factor) to the matrix of basis vectors for our lattice. This enforces that any short nonzero lattice vector must correspond to a linear combination of basis vectors which entirely cancels out the Vandermonde component. With the right size parameters, we can ensure that the number of nonzero coefficients in any such linear combination is nearly the same as the number of constraints in the PCP.
    \item We insert many relatively small reduced Vandermonde matrices (again with their entries scaled up by a large multiplicative factor) into the copy of the indicator matrix $\mat S$ contained within the lattice basis vectors. The purpose of this transformation is a bit more subtle; roughly speaking, it enforces that any short lattice vector corresponds to a linear combination of basis vectors that is at least partially ``self-consistent'' with respect to the variable assignments. See Section \ref{sec:actualconstruct} for details.
\end{enumerate}

\paragraph{Hadamard Matrices.}
The second gadget is a matrix where every pair of rows is orthogonal (in the $\ell_2$ norm), and all of the entries are integers of the same magnitude. These conditions will allow us to apply manipulations based on the relationship between the $\ell_2$ norm and the $\ell_p$ norm with $p > 2$.

\begin{definition}[Hadamard Matrix]
    Define $\mat H_0 \coloneqq \left[1\right]$, and inductively set
    \[\mat H_{i + 1} \coloneqq \begin{bmatrix}
        \mat H_i & -\mat H_i \\
        \mat H_i & \mat H_i
    \end{bmatrix}\]
    For $n$ a power of two, the $n \times n$ Hadamard matrix is $\mat H_{\log_2n}$.
\end{definition}

In our construction, we use a \emph{block diagonal} set of Hadamard matrices concatenated to the set of basis vectors. To gain some intuition, let $\mat Y \coloneqq \mat I_m \otimes \mat H_{\log_2 n}$ and consider the $\ell_p$ norm of linear combinations $\mat v \mat Y$.

\begin{enumerate}
    \item For two vectors $\mat u, \mat v$ with all entries coming from the set $\{-1, 0, 1\}$, if $\|\mat u\|_0 > \|\mat v\|_0$, then $\|\mat u \mat Y\|_2 > \|\mat v\mat Y\|_2$. This just follows from the orthogonality of the rows in $\mat Y$. In fact, we can make similar statements for $\ell_p$ with $p > 2$, assuming certain properties about $\mat u$ and $\mat v$. With the right size parameters, the Hadamard gadgets will allow us to argue that any short nonzero vector in our final lattice must come from a linear combination of basis vectors where the number of nonzero coefficients is not too much larger than the number of constraints in the PCP. (See Section \ref{sec:actualconstruct} for details.) This complements the effect of the reduced Vandermonde matrices discussed earlier.
    \item In higher $\ell_p$ norms, $\mat Y$ exhibits a certain ``anti-concentration'' property. Consider two vectors $\mat u$ and $\mat v$ with all entries coming from the set $\{-1, 0, 1\}$, such that $\|\mat u\|_0 = \|\mat v\|_0$. If the nonzero entries in $\mat u$ correspond to a relatively small number of Hadamard blocks, but the nonzero entries in $\mat v$ are relatively spread out across the blocks, then we will have $\|\mat u\mat Y\|_p \gg \|\mat v\mat Y\|_p$ for all $p > 2$. (A concrete analysis of this property is given in Lemma \ref{lem:numbertests}.) This is useful because we can force short lattice vectors to come from a linear combination of basis vectors that is spread across many of the different PCP constraints.
\end{enumerate}

Below, we give an illustration showing the anti-concentration property. \\

\noindent
\begin{minipage}{0.1\textwidth}
\end{minipage}%
\hfill%
\begin{minipage}{0.3\textwidth}
\includegraphics[width=\linewidth]{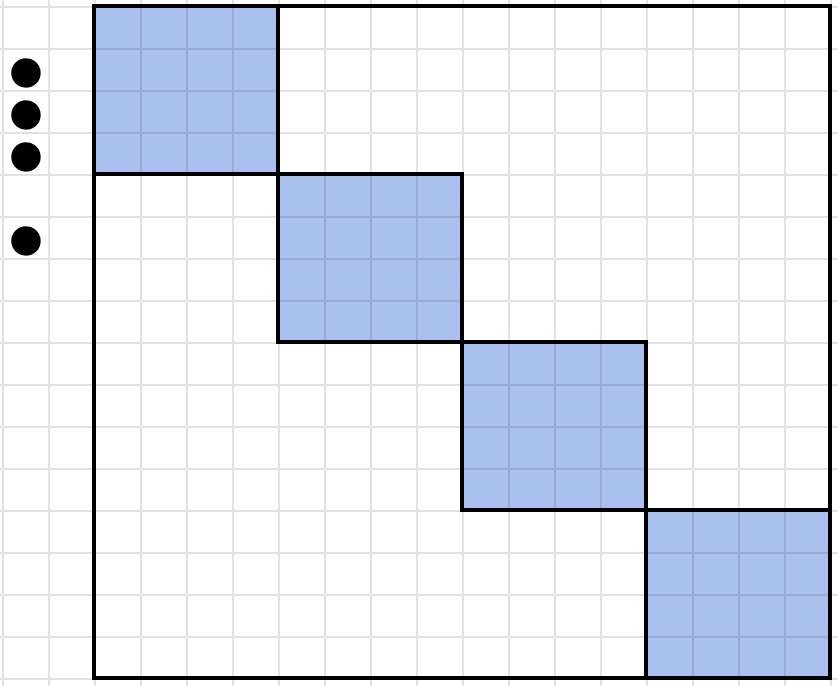}
\\\\
\includegraphics[width=\linewidth]{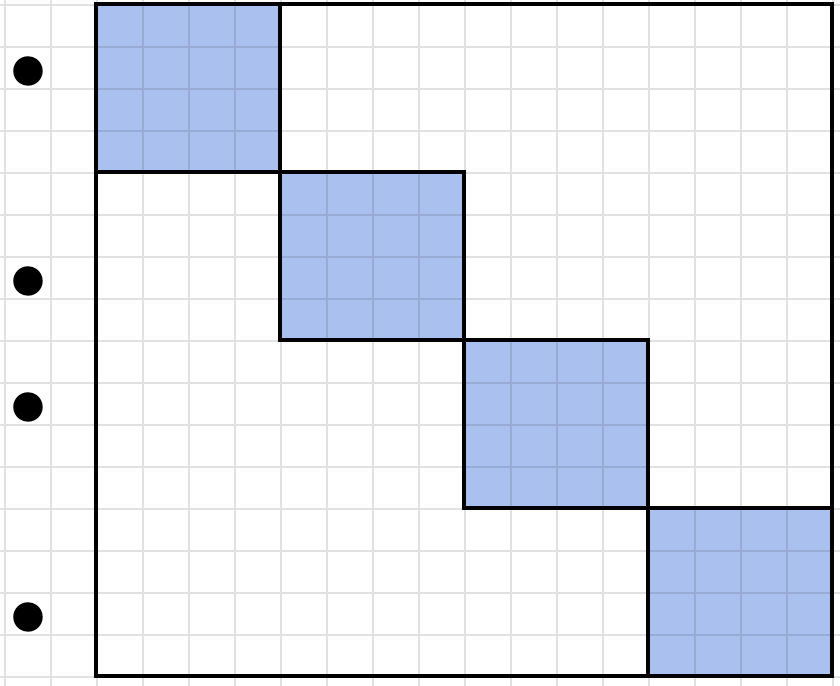}
\end{minipage}%
\hfill%
\begin{minipage}{0.5\textwidth}
(Top) The linear combination corresponding to $\mat u$, with $\|\mat u\|_0 = 4$. (Bottom) The linear combination corresponding to $\mat v$, with $\|\mat v\|_0 = 4$. In each diagram, the rows of the matrix $\mat Y$ participating in the linear combination are indicated by black dots.\\\\
We have that $\mat u \mat Y$ is a vector containing mostly zeros, but the nonzero entries will (on average) be much greater than 1. In higher $\ell_p$ norms, the length contribution from these entries is amplified significantly. In contrast, $\mat v\mat Y$ is a vector whose entries are all either $-1$ or $1$; while the number of nonzero entries in $\mat v \mat Y$ is much larger than in $\mat u \mat Y$, all of them are small, so $\ell_p$ length amplification doesn't occur.
\end{minipage}%
\hfill%
\begin{minipage}{0.1\textwidth}
\end{minipage}%
\hfill

\section{Reducing from a PCP Instance to an SVP Instance}
\label{sec:actualconstruct}

We start with a regularized PCP that has $M$ constraints and $N$ variables over alphabet $\Sigma$, where each constraint depends on exactly $q$ variables, each variable appears in exactly $d$ constraints, and the soundness parameter is $s$. Assume that $q \leq O(\log M)$, that $(1/s)^{1/q} \geq (\log M)^{\omega(1)}$, and that the alphabet size $\vert \Sigma \vert$ is a power of two. Two additional properties hold automatically:
\begin{enumerate}
    \item $q$ must be at least 2. Otherwise, it is trivial to find a maximally satisfying assignment for the PCP.
    \item $d$ must be at least $1/(sq) \geq (\log M)^{\omega(1)}$. This is because we can always find a set of $M/(dq)$ variable-disjoint constraints, and these are mutually satisfied by taking a locally satisfying assignment for each. If $d < 1/(sq)$, then the number of satisfied constraints is strictly more than $M/(1/s) = M s$, violating soundness.
\end{enumerate}
Throughout, fix a prime $a = \Theta((M\vert\Sigma\vert^q)^2)$.

We construct the matrix of basis vectors for a lattice by applying two procedures, each of which introduces a distinct set of gadgets. (The construction is agnostic to the specific choice of $\ell_p$ norm so long as $p > 2$.) We encourage the reader to adopt the following perspective: Each individual gadget converts many of the previously ``short problematic'' vectors into vectors with a large $\ell_p$ norm. When all gadgets are considered together, we show that we can map all of the remaining short vectors to assignments for the PCP which violate soundness. We will be able to go in the other direction, too: a satisfying assignment for the PCP can be mapped to a nonzero vector with short $\ell_p$ norm in the final lattice.

\paragraph{Manipulating the Indicator Matrix.}

Our first procedure constructs a matrix $\mat P$ based on the indicator matrix for the PCP by inserting reduced Vandermonde matrices in place of each nonzero entry.

\begin{enumerate}
    \item Let $\mat S \in \{0, 1\}^{M|\Sigma|^q \times N|\Sigma|}$ be the indicator matrix for the PCP. Assume that $\vert \Sigma\vert$ is a power of two.
    \item Construct matrix $\mat P \in \mathbb{Z}^{M|\Sigma|^q \times Mq|\Sigma|/(\log^3 M)}$ as follows. Start with an empty matrix $\mat P$, and for each column $\mat c$ of $\mat S$:
    \begin{enumerate}
        \item Let $\mat V \in \mathbb{Z}^{(a-1) \times d/(\log^3 M)}$ be an $(a, d/(\log^3 M))$ reduced Vandermonde matrix.
        \item Define a matrix $\mat V' \in \mathbb{Z}^{M\vert\Sigma\vert^q \times d/(\log^3 M)}$ as follows:
        \begin{enumerate}
            \item If $\mat c_{(t, \alpha_1, \ldots, \alpha_q)} = 0$, then $\mat V'_{(t, \alpha_1, \ldots, \alpha_q)} = \mat 0^{d/(\log^3 M)}$.
            \item If $\mat c_{(t, \alpha_1, \ldots, \alpha_q)} = 1$, we set $\mat V'_{(t, \alpha_1, \ldots, \alpha_q)}$ to be a distinct row\footnote{We only need $\mat V'_{(t, \alpha_1, \ldots, \alpha_q)}$ to be distinct with respect to the other rows in the current matrix $\mat V'$, not with respect to the rows in matrices $\mat V'$ constructed previously.} of $\mat V$, scaled up by a multiplicative factor of $(M\vert\Sigma\vert^q/s)^2$.
        \end{enumerate}
        \item Insert $\mat V'$ as a new set of columns in $\mat P$.
    \end{enumerate}
    
    Recall that because the PCP is regular, we must have $Nd = Mq$. So the final width of $\mat P$ is indeed $N\vert\Sigma\vert \cdot d/(\log^3 M) = Mq\vert\Sigma\vert/(\log^3 M)$. The construction is well-defined because $d \geq (\log M)^{\omega(1)}$, meaning $d/(\log^3 M) \geq (\log M)^{\omega(1)}$.\\
\end{enumerate}

\noindent\begin{minipage}{0.2\textwidth}
\includegraphics[width=\linewidth]{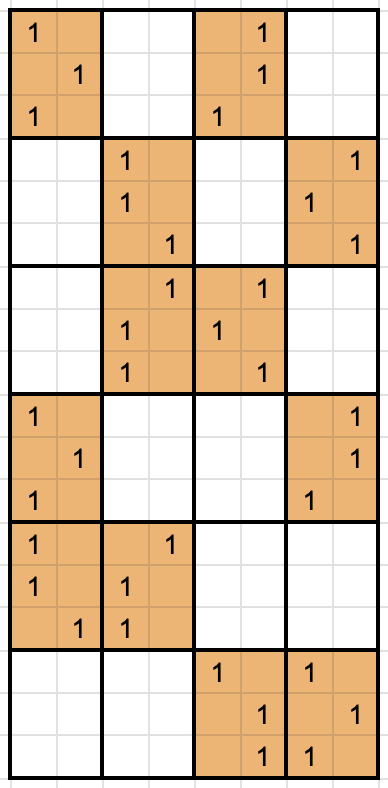}
\end{minipage}%
\hfill%
\begin{minipage}{0.02\textwidth}
\end{minipage}%
\hfill%
\begin{minipage}{0.45\textwidth}
\includegraphics[width=\linewidth]{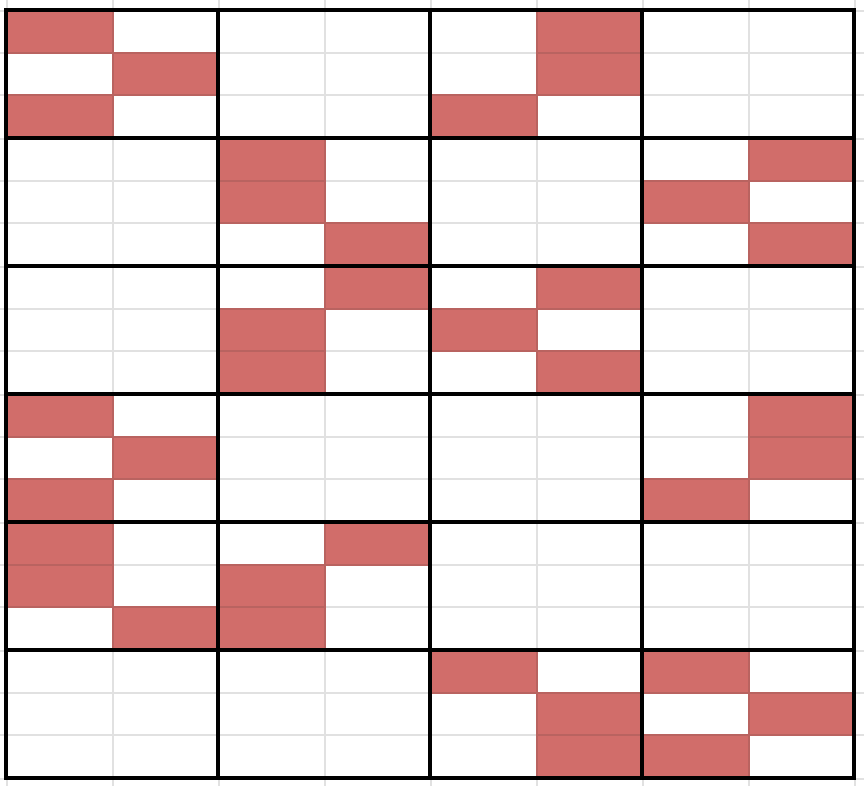}
\end{minipage}%
\hfill%
\begin{minipage}{0.02\textwidth}
\end{minipage}%
\hfill%
\begin{minipage}{0.28\textwidth}
(Left) The indicator matrix $\mat S$ for a PCP with $N = 4$, $M = 6$, $\vert \Sigma\vert = 2$, and $q = 2$. As before, we omit zero rows and only write the ``1'' entries. (Right) The corresponding matrix $\mat P$. Every red block is a row taken from a reduced Vandermonde matrix.
\end{minipage}%
\hfill

\paragraph{Why is this useful?}
Consider an arbitrary vector $\mat v \in \mathbb{Z}^{M\vert\Sigma\vert^q}$, and let $\mathcal{S}$ be the set of rows from $\mat P$ having a nonzero coefficient in the linear combination $\mat v \mat P$. Because every Vandermonde entry is scaled up by a ``very large'' multiplicative factor, we have that $\mat v \mat P$ is either zero or has a large $\ell_p$ norm, for any choice of $p \geq 1$. In the case that $\mat v \mat P$ is zero, an application of Corollary \ref{cor:Vandermonde} on the reduced Vandermonde gadgets shows that a certain gap property holds: \emph{every} set of columns corresponding to a (variable, assignment) tuple in $\mat P$ satisfies one of the following properties.
\begin{enumerate}
    \item None of the basis vectors in $\mathcal{S}$ have a nonzero entry in that set of columns.
    \item At least $d/(\log^3 M)$ of the basis vectors in $\mathcal{S}$ have nonzero entries in that set of columns.
\end{enumerate}

We show in Lemma \ref{lem:numberpairs} that this gap property ensures any short lattice vector resulting from a linear combination of (say) $M$ basis vectors cannot have those basis vectors ``indicating'' more than $f \cdot N$ (variable, assignment) tuples, for some small slack factor $f$. The conversion from a bound in terms of $M$ to a bound in terms of $N$ makes critical use of the fact that the PCP is regular.

\paragraph{Appending More Gadgets.}
In the second procedure, we append Hadamard matrices and another reduced Vandermonde matrix to $\mat P$. The total number of basis vectors does not change; instead, the width of each basis vector increases significantly.

\begin{enumerate}
    \item Let $\mat X \in \mathbb{Z}^{M\vert\Sigma\vert^q \times M/(\log^3 M)}$ be the first $M\vert\Sigma\vert^q$ rows of an $(a, M/(\log^3 M))$ reduced Vandermonde matrix, but where each entry is scaled up by a multiplicative factor of $(M\vert\Sigma\vert^q/s)^2$. As before, this is well-defined because $d \geq (\log M)^{\omega(1)}$, meaning $d/(\log^3 M) \geq (\log M)^{\omega(1)}$.
    \item Let $\mat H \in \{\pm 1\}^{\vert\Sigma\vert^q \times \vert\Sigma\vert^q}$ be a Hadamard matrix. Define $\mat Y = \mat I_M \otimes \mat H$, where $\mat I_M$ is the $M \times M$ identity matrix and $\otimes$ is the Kronecker product. Each copy of the Hadamard matrix should align with a set of rows in $\mat P$ representing the candidate assignments for a single constraint.
    \item Define $\mat G \coloneqq \big[\mat P \| \mat X \| \mat Y\big]$, and then delete all rows of $\mat G$ which are entirely zero when restricted to the submatrix $\mat P$. All remaining rows in $\mat G$ are those which correspond to a \emph{satisfying assignment} for some constraint. The height $M'$ of $\mat G$ is upper bounded as $M' \leq M\vert\Sigma\vert^q$, and the width $N'$ of $\mat G$ is exactly $N' \coloneqq Mq\vert\Sigma\vert/(\log^3 M) + M/(\log^3 M) + M\vert\Sigma\vert^q$. Because $2 \leq q \leq O(\log M)$, we have $N' = \Theta(M\vert\Sigma\vert^q)$.
\end{enumerate}

Below, we give an illustration of the matrix $\mat G$. Reduced Vandermonde gadgets are marked in red, and Hadamard gadgets are marked in blue. Because of the row deletion step applied to $\mat G$, all of the Hadamard matrices (which were orignaly square) are now wide. Asymptotically, the main width contribution is from the Hadamards.\\

\includegraphics[width=\textwidth-3em]{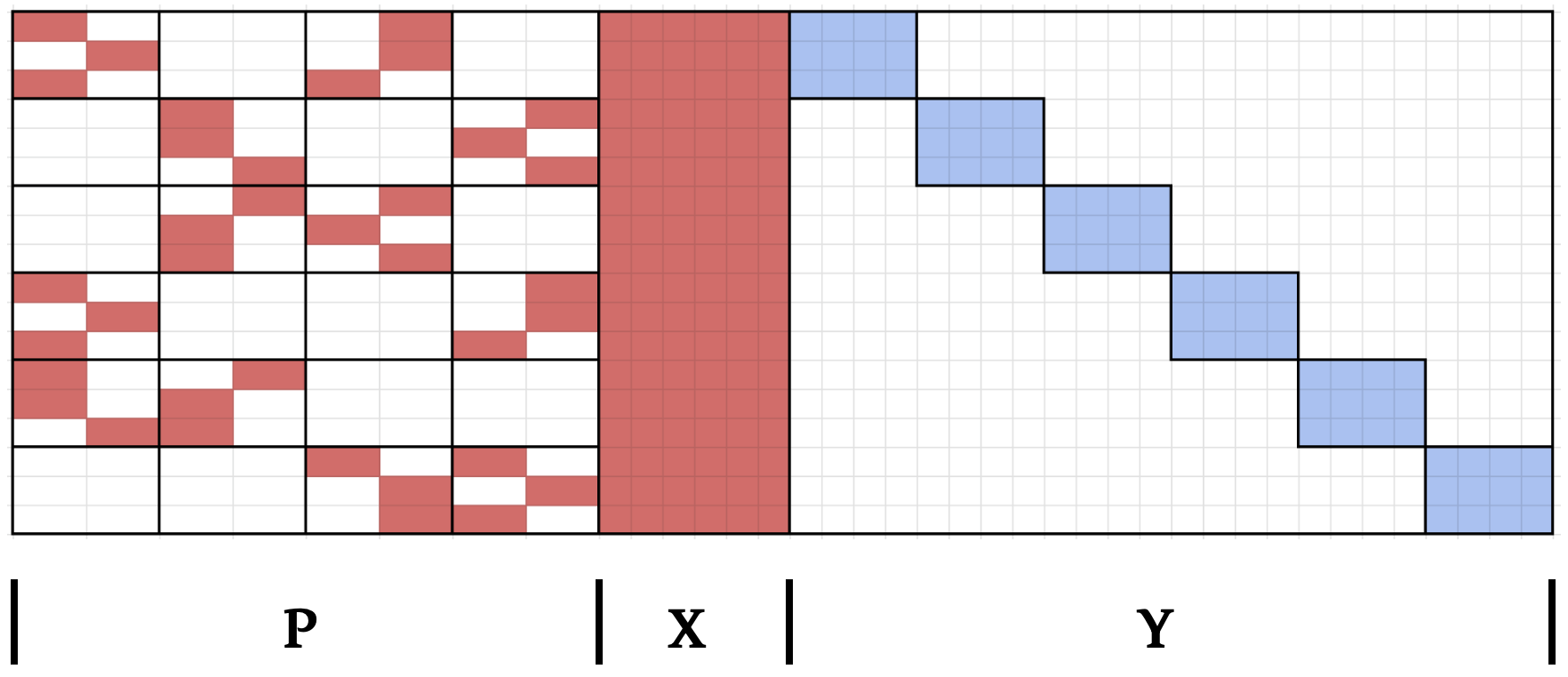}

\paragraph{Considering all Gadgets Together.} Now we examine the behavior of the new gadgets when we consider short vectors in the lattice spanned by the rows of $\mat{G}$. The purpose of $\mat X$ is to ensure that any short lattice vector corresponds to a linear combination of at least $M/(\log^3 M)$ basis vectors. As we discuss in Section \ref{sec:containscompleteness}, the number of basis vectors used to make a short vector in the completeness case (that is, when the PCP is satisfiable) will be at most $M$, but at least $M/(\log^3 M)$. Thus $\mat X$ ensures that (regardless of whether the PCP is satisfiable) any short lattice vector must use nearly the same number of basis vectors as for the completeness case.

Using the lower bound implied by $\mat X$, combined with the properties of the block Hadamard gadget $\mat Y$, we show in Lemma \ref{lem:numbertests} that any short lattice vector must correspond to a linear combination that takes basis vectors from at least $M/f'$ different PCP constraints, for some small slack factor $f'$. Recall that every row of $\mat G$ corresponds to a satisfying assignment for some constraint, so we deduce that \emph{any short lattice vector must pick at least one satisfying assignment for at least $M/f'$ distinct constraints.}

In Lemma \ref{lem:implication}, we use a subselection argument (which crucially uses the regularity of the PCP) to show that any linear combination of basis vectors which
\begin{enumerate}
    \item Indicates at most $f \cdot N$ (variable, assignment) tuples, and
    \item Picks at least one satisfying assignment from at least $M/f'$ distinct constraints
\end{enumerate}
implies the existence of an assignment to the PCP which satisfies at least an $(f \cdot f')^{-O(q)}$ fraction of constraints. Assuming that the slack factors $f$ and $f'$ are small enough with respect to $(1/s)^{1/q}$, this is enough to violate the soundness of the PCP. And by the discussion of the gadgets given above, we know that \emph{every} short lattice vector must correspond to a linear combination that satisfies both of the conditions. In other words, the existence of a short lattice vector implies the satisfiability of the PCP. We now elaborate with formal statements and full proofs.

\subsection{Formal Statement of the SVP Result} \label{sec:containscompleteness}

Fix any constant $p > 2$, or set $p = \infty$. Our SVP problem will be with respect to the $M'$ dimensional lattice $\mathcal{L} \coloneqq \{\mat v \mat G : \mat v \in \mathbb{Z}^{M'}\}$, where the $M'$ by $N'$ matrix $\mat G$ is constructed in Section \ref{sec:actualconstruct}. The approximation factor is $\gamma = (1/s)^{(1/2 - 1/p)/(25q)}$,\footnote{For $p = \infty$ this becomes $(1/s)^{1/(50q)}$.} and we measure vector length using the $\ell_p$ norm. Recall that that the PCP from which we constructed $\mat G$ satisfies $2 \leq q \leq O(\log M)$ and $(1/s)^{1/q} \geq (\log M)^{\omega(1)}$, which implies $\gamma = (\log M)^{\omega(1)}$. Because $M = 2^r$ (recall that $r$ is the number of random bits used by the verifier), we can rewrite these inequalities as $2 \leq q \leq O(r)$ and $(1/s)^{1/q} \geq r^{\omega(1)}$. It's clear from Section \ref{sec:actualconstruct} that the construction proceeds deterministically in $(2^r\vert\Sigma\vert^q)^{O(1)}$ time. We claim:
\begin{claim}[Completeness] \label{claim:completeness}
    If the PCP has a satisfying assignment, then assuming $M$ is sufficiently large, there exists a nonzero vector $\mat w \in \mathcal{L}$ with $\|\mat w\|_p \leq (N')^{1/p}$.\footnote{In the case of $p = \infty$, we have $(N')^{1/p} = 1$.}
\end{claim}
\begin{claim}[Soundness] \label{claim:soundness}
    If the PCP does not have a satisfying assignment, then assuming $M$ is sufficiently large, every nonzero vector $\mat w \in \mathcal{L}$ has $\|\mat w\|_p > \gamma (N')^{1/p}$.
\end{claim}

From these claims we deduce the main theorem. 

\begin{theorem}\label{thm:mainresultformal}
    Suppose there exists a PCP for SAT instances of size $n$ that uses $r$ random bits to read $q \leq O(r)$ locations in a proof over alphabet $\Sigma$ and has soundness parameter $s$. Suppose further that every proof location is read with equal probability, that $(1/s)^{1/q} \geq r^{\omega(1)}$, and that the PCP can be constructed deterministically in time $T(n)$. Then for all constants $p > 2$ (and for $p = \infty$), there exists a $T(n) + (2^r\vert\Sigma\vert^q)^{O(1)}$ time deterministic reduction from SAT instances of size $n$ to instances of $\gamma$-GapSVP$_p$ on lattices of dimension $O(2^r\vert\Sigma\vert^q)$, where $\gamma = (1/s)^{(1/2-1/p)/(25q)}$.
\end{theorem}

\section{Proof of Claim \ref{claim:completeness} (Completeness)}
Here we show that a satisfying assignment for the PCP maps to a short lattice vector. We will use the following proposition to help zero out the reduced Vandermonde gadgets, which ensures that our lattice vector will not have any large-magnitude entries.

\begin{proposition} \label{lem:SmallCombination}
    Let $\mat M \in \mathbb{Z}^{n \times O(n/\log^2 n)}$ be a matrix with entries in $\{-n^{O(1)}, \ldots, n^{O(1)}\}$. Then assuming $n$ is sufficiently large, there exists a nonzero vector $\mat v \in \{-1, 0, 1\}^{n}$ with $\mat v \mat M = \mat 0$.
\end{proposition}

\begin{proof}
    There are $2^{n}$ vectors $\mat w \in \{0, 1\}^{n}$, and there are $n^{O(n/\log^2 n)} = 2^{O(n/\log n)}$ possible values for $\mat w \mat M$. By the pigeonhole principle, there must exist two distinct vectors $\mat w', \mat w''$ with $\mat w' \mat M = \mat w'' \mat M$. Take $\mat v = \mat w' - \mat w''$.
\end{proof}

Now let $\sigma$ be a satisfying assignment for the PCP, and let $\mat G'$ be the submatrix of $\mat G$ induced by the $M$ rows corresponding to $\sigma$. ($\mat G'$ consists of exactly one row per PCP constraint.) We can write $\mat G' = [\mat P' \| \mat X' \| \mat Y']$, where $\mat P', \mat X', \mat Y'$ are the corresponding row induced submatrices of $\mat P, \mat X, \mat Y$.

Observe that the nonzero entries in $\mat G'$ only touch the columns of $\mat P$ corresponding to one consistent assignment for each variable, which makes for a total of $Mq/(\log^3 M)$ columns. The submatrix $\mat X$ has width $M/(\log^3 M)$, so the nonzero entries in $\mat G'$ in fact touch $Mq/(\log^3 M) + M/(\log^3 M) = O(M/(\log^2 M))$ columns across both $\mat P$ and $\mat X$, where we used that $q \leq O(\log M)$. Thus we can apply Proposition \ref{lem:SmallCombination} to guarantee the existence of a nonzero vector $\mat v \in \{-1, 0, 1\}^{M}$ with $\mat v \mat P' = \mat 0$ and $\mat v \mat X' = \mat 0$.

The above implies that all of the nonzero entries in $\mat v \mat G'$ will come from $\mat v \mat Y'$. Due to the structure of $\mat Y$ and the fact that $\mat G'$ consists of exactly one row per PCP constraint, every column of $\mat Y'$ has exactly one nonzero entry. This implies $\|\mat v \mat Y'\|_\infty = 1$ and thus $\|\mat v \mat G'\|_\infty = 1$. By padding $\mat v$ with zeros we get a vector $\mat v'$ such that $\|\mat v' \mat G\|_{\infty} = 1$. Since the width of $\mat G$ is defined as $N'$, this implies that $\|\mat v' \mat G\|_p \leq (N')^{1/p}$, and we can take $\mat w = \mat v' \mat G$ as our short lattice vector.

\section{Proof of Claim \ref{claim:soundness} (Soundness)} \label{sec:proof}

Here we prove that, if the PCP does not have a satisfying assignment, every nonzero vector $\mat w \in \mathcal{L}$ will have a relatively large $\ell_p$ norm. Most of the work is dedicated to ruling out certain nonzero lattice points $\mat w \in \mathcal{L}$, based on the structure of the corresponding vectors $\mat v$ with $\mat v \mat G = \mat w$.

The first lemma\footnote{This lemma is also true for the case that $p = 2$.} shows that if $\mat v$ has too few nonzero entries, then $\|\mat v\mat G\|_p$ must be large. We make use of the reduced Vandermonde gadget $\mat X$.

\begin{lemma} \label{lem:minnonzero}
    Assume that $\mat v \in \mathbb{Z}^{M'}$ is a nonzero vector with $\|\mat v\|_0 \leq M/(\log^3 M)$. Then assuming $M$ is sufficiently large, $\|\mat v\mat G\|_p \geq (M\vert\Sigma\vert^q/s)^2 > \gamma(N')^{1/p}$.
\end{lemma}

\begin{proof}
    Observe that $\|\mat v \mat G\|_p \geq \|\mat v \mat X\|_p$, because $\mat X$ is a column-induced submatrix of $\mat G$. Recall that $\mat X$ consists of a subset of rows taken from an $(a, M/(\log^3 M))$ reduced Vandermonde matrix, where every entry is multiplied by $(M\vert\Sigma\vert^q/s)^2$. By Corollary \ref{cor:Vandermonde}, we know that $\mat v \mat X$ must have at least one nonzero coordinate, and the smallest nonzero value possible is $(M\vert\Sigma\vert^q/s)^2$. This gives $\|\mat v\mat G\|_p \geq (M\vert\Sigma\vert^q/s)^2$, and since $q \geq 2$, $s \leq 1$, and $p > 2$, we have $(M\vert\Sigma\vert^q/s)^2 > \gamma N' = \Theta((1/s)^{(1/2 - 1/p)/(25q)}M\vert\Sigma\vert^q)$ when $M$ is sufficiently large.
\end{proof}

The next lemma\footnote{This lemma is also true for the case that $p = 2$.} gives an upper bound on $\|\mat v\|_0$. Here, we make use of the block-diagonal Hadamard gadgets in submatrix $\mat Y$.

\begin{lemma} \label{lem:maxnonzero}
    Suppose $\mat v \in \mathbb{Z}^{M'}$ is a nonzero vector with $\|\mat v\|_0 \geq M\cdot \gamma^3$. Then assuming $M$ is sufficiently large, $\|\mat v\mat G\|_p > \gamma(N')^{1/p}$.
\end{lemma}

\begin{proof}
    Because $\mat Y$ is a column-induced submatrix of $\mat G$, we have $\|\mat v\mat G\|_p \geq \|\mat v \mat Y\|_p$. Since each row of $\mat{Y}$ is a vector in $\{\pm 1\}^{\vert\Sigma\vert^q}$, and every pair of rows is orthogonal (in the $\ell_2$ norm), we have that
    \[\|\mat v \mat{Y}\|_2 = \Omega\left(\sqrt{\|\mat v\|_0\vert\Sigma\vert^q}\right) = \Omega\left(\sqrt{M \cdot \gamma^3 \cdot \vert\Sigma\vert^q}\right) = \Omega(\gamma^{3/2} \cdot \sqrt{N'}),\]
    where the last equality holds because $N' = \Theta(M\vert\Sigma\vert^q)$.
    
    Because $p > 2$, an application of Corollary \ref{cor:Holder} gives
    \begin{align*}
        \|\mat v \mat{Y}\|_p \geq & (N')^{1/p - 1/2} \cdot \|\mat v \mat{Y}\|_2 \\
        = & \Omega((N')^{1/p - 1/2} \cdot \gamma^{3/2} \cdot \sqrt{N'}) \\
        = & \Omega(\gamma^{3/2} \cdot (N')^{1/p}).
    \end{align*}
    Since $\gamma \geq (\log M)^{\omega(1)}$, we have $\Omega(\gamma^{3/2}) > \gamma$ for sufficiently large $M$, in which case $\|\mat v \mat G\|_p > \gamma(N')^{1/p}$.
\end{proof}

Later, we actually use the contrapositive of Lemmas \ref{lem:minnonzero} and \ref{lem:maxnonzero}. In particular, we know that if $\|\mat v \mat G\|_p \leq \gamma(N')^{1/p}$ and $\mat v$ is nonzero, then $M/(\log^3 M) < \|\mat v\|_0 < M \cdot \gamma^3$.

In the remaining lemmas, we need to formalize the correspondence between nonzero entries in the vector $\mat v$ and different components of $\mat G$ and of the PCP. Our first definition maps from nonzero entries of $\mat v$ to rows of $\mat G$. (Recall that each row of $\mat G$ is a basis vector of the lattice.)

\begin{definition}[Indicated Basis Vectors]
    A row vector $\mat r$ from $\mat G$ is said to be \emph{indicated by} $\mat v$ if the linear combination $\mat v \mat G$ assigns a nonzero coefficient to $\mat r$.
\end{definition}

The next definition formalizes the correspondence between $\mat v$ and different PCP constraints.

\begin{definition}[Indicated Constraint]
    We say that \emph{$\mat v$ indicates constraint $t$} if $\mat v$ indicates at least one row vector representing an assignment to constraint $t$.
\end{definition}

In the following we use the Hadamard matrices again, this time exploiting (in a formal manner) the ``anti-concentration'' property.\footnote{This property only holds for $p$ norms with $p > 2$.}

\begin{lemma} \label{lem:numbertests}
    Let $\mat v \in \mathbb{Z}^{M'}$ be a nonzero vector, and assume that $\mat v$ indicates at most $M/\gamma^{2/(1/2 - 1/p)}$ distinct constraints of the PCP. Then assuming $M$ is sufficiently large, $\|\mat v\mat G\|_p > \gamma(N')^{1/p}$.
\end{lemma}

\begin{proof}
    We can assume that $\|\mat v\|_0 \geq M/(\log^3 M)$; otherwise by Lemma \ref{lem:minnonzero}, we already have $\|\mat v \mat G\|_p \geq (M\vert\Sigma\vert^q/s)^2 > \gamma(N')^{1/p}$. As before, we can write $\|\mat v \mat G\|_p \geq \|\mat v \mat Y\|_p$.
    
    We assumed that $\mat v$ indicates at most $M/\gamma^{2/(1/2 - 1/p)}$ distinct constraints of the PCP. So by the construction of $\mat Y$, the basis vectors indicated by $\mat v$ have nonzero entries in at most $(M/\gamma^{2/(1/2 - 1/p)}) \cdot \vert\Sigma\vert^q = \Theta(N'/\gamma^{2/(1/2 - 1/p)})$ columns of $\mat Y$. Thus we can find a column-induced submatrix $\mat{\hat{Y}}$ of $\mat Y$ such that the width of $\mat{\hat{Y}}$ is at most $\Theta(N'/\gamma^{2/(1/2 - 1/p)})$, and $\|\mat v \mat Y\|_p = \|\mat v \mat{\hat{Y}}\|_p $. Since each row of $\mat{\hat{Y}}$ is a vector in $\{\pm 1\}^{\vert\Sigma\vert^q}$, and every pair of rows is orthogonal (in the $\ell_2$ norm), we have that
    \[\|\mat v \mat{\hat{Y}}\|_2 = \Omega\left(\sqrt{\|\mat v\|_0\vert\Sigma\vert^q}\right) = \Omega\left(\sqrt{\frac{M\vert\Sigma\vert^q}{\log^3 M}}\right) = \Omega(\sqrt{N'}/\log^{3/2} M),\]
    since $N' = \Theta(M\vert\Sigma\vert^q)$.
    
    Using our upper bound on the width of $\mat{\hat{Y}}$, and the fact that $p > 2$, an application of Corollary \ref{cor:Holder} gives

    \begin{align*}
        \|\mat v \mat{\hat{Y}}\|_p \geq & (N'/\gamma^{2/(1/2 - 1/p)})^{1/p - 1/2} \cdot \|\mat v \mat{\hat{Y}}\|_2 \\
        = & (N')^{1/p - 1/2} \cdot \gamma^2 \cdot \|\mat v \mat{\hat{Y}}\|_2 \\
        \geq & \Omega((N')^{1/p - 1/2} \cdot \gamma^2 \cdot \sqrt{N'}/\log^{3/2} M) \\
        = & \Omega((N')^{1/p} \cdot \gamma^2/\log^{3/2} M) \\
        > & (N')^{1/p} \cdot \gamma
    \end{align*}
    where we used that, because $\gamma = (\log M)^{\omega(1)}$, we have $\gamma^2/\log^{3/2} M > \gamma$ when $M$ is sufficiently large.

    To summarize, we have shown that $\|\mat v \mat G\|_p \geq \|\mat v \mat Y\|_p = \|\mat v \mat{\hat{Y}}\|_p > (N')^{1/p} \cdot \gamma$, completing the proof.
\end{proof}

Now our plan is to use the reduced Vandermonde matrices inserted within submatrix $\mat P$ to argue that any nonzero vector $\mat v$ with $\mat v\mat G$ being short must indicate a small number of (variable, assignment) tuples in the original PCP. First we make the notion of ``indicated assignments'' more concrete, and define a notion of multiplicity for assignments.

\begin{definition}[Indicated Assignments]
    Let $\mat v \in \mathbb{Z}^{M'}$ be any vector. Consider the row-induced submatrix $\mat{\hat{P}}$ of $\mat P$ indicated by the nonzero entries in $\mat v$. (In other words, this is the submatrix of $\mat P$ composed of all rows that are assigned a nonzero coefficient in the linear combination $\mat v \mat P$.) Every (variable, assignment) tuple $(x, \alpha)$ corresponds to a set of columns in $\mat P$, and we say that a tuple $(x, \alpha)$ is \emph{indicated by $\mat v$} iff $\mat{\hat{P}}$ has a nonzero entry in that set of columns.
\end{definition}

Note that the set of all indicated assignments may be larger than the number of variables in the PCP. We now define a notion of \emph{multiplicity}, which counts how many times $\mat v$ indicates the same assignment for the same variable.

\begin{definition}[Multiplicity]
    We say that a tuple $(x, \alpha)$ \emph{has multiplicity $h$ with respect to $\mat v$} if there exist exactly $h$ distinct nonzero entries of $\mat v$ that indicate $(x, \alpha)$.
\end{definition}

We also need a definition which counts the number of \emph{different} assignments for the same variable.

\begin{definition}[Distinct Assignment Count]
    We say that a variable $x$ \emph{has distinct assignment count $h$ with respect to $\mat v$} if there exist exactly $h$ distinct values $\alpha$ such that $\mat v$ indicates $(x, \alpha)$.
\end{definition}

Below we demonstrate that, if $\mat v \mat G$ has short $\ell_p$ norm, then $\mat v$ cannot indicate too many distinct (variable, assignment) tuples.

\begin{lemma} \label{lem:numberpairs}
    Let $\mat v \in \mathbb{Z}^{M'}$ be a nonzero vector, and assume that $\mat v$ indicates at least $N \cdot \gamma^4$ distinct (variable, assignment) tuples. Then assuming $M$ is sufficiently large, we have $\|\mat v\mat G\|_p \geq (M\vert\Sigma\vert^q/s)^2 > \gamma(N')^{1/p}$.
\end{lemma}

\begin{proof}
    We know by Lemma \ref{lem:maxnonzero} that if $\|\mat v\|_0 \geq M\cdot \gamma^3$, we automatically have $\|\mat v\mat G\|_p \geq (M\vert\Sigma\vert^q/s)^2 > \gamma(N')^{1/p}$. So assume otherwise.

    By construction, each nonzero entry of $\mat v$ will indicate exactly $q$ (variable, assignment) tuples. Thus, allowing for duplicates, $\mat v$ indicates at most $Mq \cdot \gamma^3$ (variable, assignment) tuples. By the averaging principle and the assumption that $\mat v$ indicates at least $N \cdot \gamma^4$ distinct (variable, assignment) tuples, there exists $(x, \sigma)$ with multiplicity at least one and at most
    \[\frac{Mq \cdot \gamma^3}{N \cdot \gamma^4} = d/\gamma,\]
    where we used that $Nd = Mq$.

    Now consider the column-induced submatrix $\mat{\hat{P}}$ of $\mat P$ corresponding to $(x, \sigma)$. We know that all nonzero rows of $\mat{\hat{P}}$ are distinct rows of a width $d/(\log^3 M)$ reduced Vandermonde matrix, where all entries are scaled up by a factor of $(M\vert\Sigma\vert^q/s)^2$. By the argument above, $\mat v$ corresponds to a linear combination of at least one nonzero row from $\mat{\hat{P}}$ and at most $d/\gamma$ nonzero rows from $\mat{\hat{P}}$. Using that $\gamma = (\log M)^{\omega(1)}$, we have that $d/\gamma \leq d/(\log M)^{\omega(1)} < d/\log^3 M$ for sufficiently large $M$. Therefore by Corollary \ref{cor:Vandermonde}, we have $\|\mat v\mat{\hat{P}}\|_p \geq (M\vert\Sigma\vert^q/s)^2 > \gamma(N')^{1/p}$. Because $\|\mat v\mat G\|_p \geq \|\mat v\mat{\hat{P}}\|_p$, the lemma follows.
\end{proof}

The final step is to combine all of the restrictions on $\mat v$ to show that the remaining candidate short vectors must implicate soundness of the PCP.

\begin{lemma} \label{lem:implication}
    Assume that the PCP from which we constructed $\mat G$ was not satisfiable. Then for every nonzero vector $\mat v \in \mathbb{Z}^{M'}$, it holds that $\|\mat v\mat G\|_p > \gamma(N')^{1/p}$, assuming $M$ is sufficiently large.
\end{lemma}

\begin{proof}
    Assume for contradiction that the PCP was not satisfiable, and there does exist a nonzero vector $\mat v \in \mathbb{Z}^{M'}$ with $\|\mat v\mat G\|_p \leq \gamma(N')^{1/p}$. By Lemmas \ref{lem:numbertests} and \ref{lem:numberpairs} respectively, we know:
    \begin{enumerate}
        \item $\mat v$ indicates at least $M/\gamma^{2/(1/2 - 1/p)}$ distinct constraints in the PCP.
        \item $\mat v$ indicates at most $N \cdot \gamma^4$ distinct (variable, assignment) tuples.
    \end{enumerate}
    Let $\mathcal{X}_{\textnormal{high}}$ be the set of all variables $x$ whose distinct assignment count with respect to $\mat v$ is at least $\gamma^{7 + 2/(1/2 - 1/p)}$. Our first step is to upper bound $\vert \mathcal{X}_{\textnormal{high}}\vert$.
    
    \begin{claim}
        $\vert \mathcal{X}_{\textnormal{high}} \vert \leq N/\gamma^{2 + 2/(1/2 - 1/p)}.$
    \end{claim}
    \begin{proof}
        Suppose not, then we would need for $\mat v$ to indicate at least
        \[\gamma^{7 + 2/(1/2 - 1/p)} \cdot N/\gamma^{2 + 2/(1/2 - 1/p)} = N \cdot \gamma^5\]
        distinct (variable, assignment) tuples, which violates the upper bound from Lemma \ref{lem:numberpairs} because $\gamma = (\log M)^{\omega(1)} > 1$ for sufficiently large $M$.
    \end{proof}
    
    Now we quantify the number of constraints that $\mathcal{X}_{\textnormal{high}}$ can interact with. Define $\mathcal{T}_{\textnormal{high}}$ as the set of all constraints in the PCP which are incident to at least one variable from $\mathcal{X}_{\textnormal{high}}$.
    
    \begin{claim}
        $\vert \mathcal{T}_{\textnormal{high}} \vert \leq M/\gamma^{1 + 2/(1/2 - 1/p)}.$
    \end{claim}

    \begin{proof}
        Because every variable appears in exactly $d$ constraints, we have
        \begin{align*}
            \vert \mathcal{T}_{\textnormal{high}}\vert \leq & d\vert \mathcal{X}_{\textnormal{high}}\vert \\
            \leq & Nd/\gamma^{2 + 2/(1/2 - 1/p)} \\
            = & Mq/\gamma^{2 + 2/(1/2 - 1/p)} \\
            \leq & M/\gamma^{1 + 2/(1/2 - 1/p)},
        \end{align*}
        where we used that $Nd = Mq$ and $q = O(\log M) < \gamma = (\log M)^{\omega(1)}$ for sufficiently large $M$.
    \end{proof}
    
    Let $\mathcal{T}_{\mat v}$ be the set of constraints indicated by $\mat v$, and define $\mathcal{T}_{\mat v, \textnormal{low}}$ as $\mathcal{T}_{\mat v}$ with the constraints in $\mathcal{T}_{\textnormal{high}}$ removed. We now argue that a significant fraction of the constraints must be left over.

    \begin{claim}
        $\vert \mathcal{T}_{\mat v, \textnormal{low}}\vert \geq M/\gamma^{1 + 2/(1/2 - 1/p)}.$
    \end{claim}

    \begin{proof}
        Recall that we must have $\vert \mathcal{T}_{\mat v}\vert \geq M/\gamma^{2/(1/2 - 1/p)}$ by Lemma \ref{lem:numbertests}. We have
        \begin{align*}
            \vert\mathcal{T}_{\mat v, \textnormal{low}}\vert \geq & \vert \mathcal{T}_{\mat v}\vert - \vert \mathcal{T}_{\textnormal{high}} \vert \\
            \geq & M/\gamma^{2/(1/2 - 1/p)} - M/\gamma^{1 + 2/(1/2 - 1/p)} \\
            \geq & 2 \cdot M/\gamma^{1 + 2/(1/2 - 1/p)} - M/\gamma^{1 + 2/(1/2 - 1/p)} \\
            \geq & M/\gamma^{1 + 2/(1/2 - 1/p)}
        \end{align*}
        where we used that $\gamma \geq 2$ when $M$ is sufficiently large.
    \end{proof}

    Now we derive a contradiction by implicating the soundness of the PCP. For each variable $x$ in the PCP, pick a tuple $(x, \sigma)$ uniformly at random from the set of all tuples indicated by $\mat v$ that involve $x$. (If $\mat v$ has no tuples for the variable $x$, assign $x$ to an arbitrary alphabet symbol.) Denote this random assignment as $\mathcal{R}$.

    \begin{claim}
        The expected number of constraints satisfied by assignment $\mathcal{R}$ is at least $\vert \mathcal{T}_{\mat v, \textnormal{low}}\vert /\gamma^{7q + 2q/(1/2 - 1/p)}$.
    \end{claim}

    \begin{proof}
        Observe that each constraint $t$ in $\mathcal{T}_{\mat v, \textnormal{low}}$ is by definition \emph{not} incident to a variable in $\mathcal{X}_{\textnormal{high}}$, meaning that each of its variables has at most $\gamma^{7 + 2/(1/2 - 1/p)}$ distinct assignments indicated by $\mat v$. (By definition, each incident variable must also have at least one assignment indicated by $\mat v$.) Now even if $\mat v$ only indicates one basis vector corresponding to $t$, the probability that the exact matching assignment is selected is $1/\gamma^{7q + 2q/(1/2 - 1/p)}$. (Keep in mind that, by construction of $\mat G$, every row that $\mat v$ can pick for a constraint will correspond to a \emph{satisfying} assignment.) By linearity of expectation, the expected number of satisfied constraints is at least $\vert \mathcal{T}_{\mat v, \textnormal{low}}\vert /\gamma^{7q + 2q/(1/2 - 1/p)}$.
    \end{proof}

    Using our lower bound on $\vert \mathcal{T}_{\mat v, \textnormal{low}}\vert$, we can re-write this expectation as
    \begin{align*}
        \vert \mathcal{T}_{\mat v, \textnormal{low}}\vert /\gamma^{7q + 2q/(1/2 - 1/p)} \geq & M/(\gamma^{1 + 2/(1/2 - 1/p)} \cdot \gamma^{7q + 2q/(1/2 - 1/p)}) \\
        \geq & M/(\gamma^{7q + 2q/(1/2 - 1/p)} \cdot \gamma^{7q + 2q/(1/2 - 1/p)}) \\
        \geq & M/\gamma^{20q/(1/2 - 1/p)} \\
        = & M/(1/s)^{((1/2 - 1/p)/(25q)) \cdot (20q/(1/2 - 1/p))} \\
        = & M \cdot s^{((1/2 - 1/p)/(25q)) \cdot (20q/(1/2 - 1/p))} \\
        = & M \cdot s^{4/5},
    \end{align*}
    where we used that $\gamma = (1/s)^{(1/2 - 1/p)/(25q)}$. By picking the best choice of randomness, we get a deterministic assignment which violates soundness.
\end{proof}

\noindent\begin{minipage}{0.3\textwidth}
\includegraphics[width=0.8\linewidth]{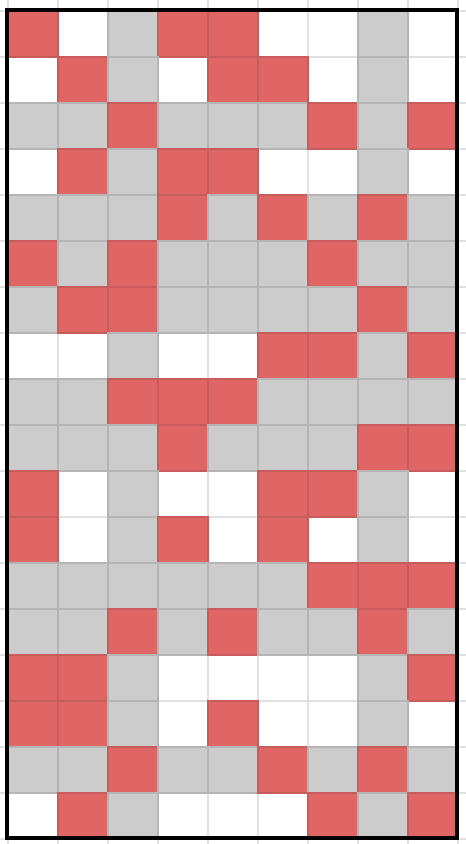}
\end{minipage}%
\hfill%
\begin{minipage}{0.02\textwidth}
\end{minipage}%
\hfill%
\begin{minipage}{0.66\textwidth}
An example PCP with $N = 9$, $M = 18$, $q = 3$, $d = 6$. Every column represents a variable, and every row represents a constraint. In this example, we have two variables in $\mathcal{X}_{\textnormal{high}}$; their columns are marked in gray. We also mark every constraint in $\mathcal{T}_{\textnormal{high}}$ using gray. We can afford to ``throw away'' all constraints in $\mathcal{T}_{\textnormal{high}}$ because $\vert\mathcal{T}_{\textnormal{high}}\vert \ll \vert\mathcal{T}_{\mat v}\vert$. (This effect becomes more pronounced as the PCP gets larger.)
\end{minipage}%
\hfill
\vspace{10px}

\section{Putting it all Together}
\label{sec:puttingittogether}

In this section we derive Corollaries \ref{cor:NPhard} and \ref{cor:sliding}. Recall the PCP theorem of \cite{dinur1999pcp}, the Sliding Scale Conjecture \cite{bellare1993efficient}, and the regularization technique of Hirahara and Moshkovitz \cite{hirahara2023regularization}.

\paragraph{Theorem \ref{thm:PCPbeta} \textnormal{(Restated, from~\cite{dinur1999pcp})}.}
    \emph{For every constant $\varepsilon > 0$, there exists a PCP for SAT instances of size $n$ that uses $r = O(\log n)$ random bits to make $q = O(1)$ queries to a proof over alphabet $\Sigma$, where the alphabet size is $\vert\Sigma\vert = 2^{\Theta(\log^{1 - \varepsilon} n)}$ and the soundness parameter is $s = 2^{-\Theta(\log^{1 - \varepsilon} n)}$. The PCP can be constructed deterministically in $n^{O(1)}$ time.}

\begin{conjecture}[Sliding Scale Conjecture, from~\cite{bellare1993efficient}]\label{conj:sliding}
    There exists a PCP for SAT instances of size $n$ that uses $r = O(\log n)$ random bits to make $q = O(1)$ queries to a proof over alphabet $\Sigma$, where the alphabet size is $\vert\Sigma\vert = n^{\Theta(1)}$ and the soundness parameter is $s = n^{-\Theta(1)}$. The PCP can be constructed deterministically in $n^{O(1)}$ time.
\end{conjecture}

\paragraph{Theorem \ref{thm:regtechnique} \textnormal{(Restated, from~\cite{hirahara2023regularization})}.}
    \emph{Assume that we have a PCP verifier $\mathcal{V}$ which uses $r$ random bits to make $q$ queries to a proof over alphabet $\Sigma$ and has soundness $s$. Assume further that $s \leq \min(1/(3q), 1/\vert\Sigma\vert^c)$, where $c$ is a constant with $0 < c \leq 1$. Then we can construct a new PCP verifier $\mathcal{V}'$ from $\mathcal{V}$ deterministically in $\vert \mathcal{V}\vert^{O(1)}$ time, such that the following holds. $\mathcal{V}'$ uses $r + O(\log(1/s))$ random bits to make $q^{O(1)}$ queries to a proof over the same alphabet $\Sigma$ and has soundness $s^{\Theta(1)}$. Additionally, every proof symbol is read for exactly $d = (q/s)^{\Theta(1)}$ choices of randomness.}\\

Now recall our main theorem:

\paragraph{Theorem \ref{thm:mainresultformal} \textnormal{(Restated)}.}
    \emph{Suppose there exists a PCP for SAT instances of size $n$ that uses $r$ random bits to read $q \leq O(r)$ locations in a proof over alphabet $\Sigma$ and has soundness parameter $s$. Suppose further that every proof location is read with equal probability, that $(1/s)^{1/q} \geq r^{\omega(1)}$, and that the PCP can be constructed deterministically in time $T(n)$. Then for all constants $p > 2$ (and for $p = \infty$), there exists a $T(n) + (2^r\vert\Sigma\vert^q)^{O(1)}$ time deterministic reduction from SAT instances of size $n$ to instances of $\gamma$-GapSVP$_p$ on lattices of dimension $O(2^r\vert\Sigma\vert^q)$, where $\gamma = (1/s)^{(1/2-1/p)/(25q)}$.}\\

We first prove the unconditional NP hardness result.

\paragraph{Corollary \ref{cor:NPhard} \textnormal{(Restated)}.}
    \emph{For all constants $\varepsilon > 0$ and $p > 2$, $\gamma$-GapSVP$_p$ on $n$ dimensional lattices is NP hard, where $\gamma = 2^{\log^{1 - \varepsilon} n}$.}

\begin{proof}
    Set $\varepsilon' = \varepsilon/2$. The PCP from Theorem \ref{thm:PCPbeta} meets the preconditions of Theorem \ref{thm:regtechnique}, so together they imply that there exists a PCP for SAT instances of size $n$ that uses $r = O(\log n)$ random bits to read $q = O(1)$ locations in a proof over alphabet $\Sigma$, where the alphabet size is $\vert\Sigma\vert = 2^{\Theta(\log^{1 - \varepsilon'}n)}$, the soundness parameter is $s = 2^{-\Theta(\log^{1 - \varepsilon'} n)}$, and every proof symbol is read for exactly $d = 2^{\Theta(\log^{1 - \varepsilon'} n)}$ choices of randomness. Clearly $q \leq O(r)$ in this case, and $(1/s)^{1/q} = 2^{\Theta(\log^{1 - \varepsilon'} n)} \geq r^{\omega(1)} = (\log n)^{\omega(1)}$. We also have that $2^r\vert\Sigma\vert^q = n^{O(1)}$. Putting all of this together, we can apply Theorem \ref{thm:mainresultformal} to show that $\gamma$-GapSVP$_p$ is NP hard on lattices of dimension $n^{O(1)}$, where $\gamma = 2^{\Theta(\log^{1 - \varepsilon'}n)}$. Rescaling the dimension of the lattice and using $\varepsilon$ in place of $\varepsilon'$ yields the corollary.
\end{proof}

The proof of the conditional NP hardness result is nearly the same.

\paragraph{Corollary \ref{cor:sliding} \textnormal{(Restated)}.}
    \emph{Assuming the Sliding Scale Conjecture, $\gamma$-GapSVP$_p$ on $n$ dimensional lattices is NP hard for all constants $p > 2$ (and for $p = \infty$), where $\gamma = n^c$ for a constant $c > 0$ that depends on $p$.}

\begin{proof}
    The PCP from Conjecture \ref{conj:sliding} meets the preconditions of Theorem \ref{thm:regtechnique}, so together they imply that there exists a PCP for SAT instances of size $n$ that uses $r = O(\log n)$ random bits to read $q = O(1)$ locations in a proof over alphabet $\Sigma$, where the alphabet size is $\vert\Sigma\vert = n^{\Theta(1)}$, the soundness parameter is $s = n^{-\Theta(1)}$, and every proof symbol is read for exactly $d = n^{\Theta(1)}$ choices of randomness. As before, $q \leq O(r)$, and $(1/s)^{1/q} = n^{\Theta(1)} \geq r^{\omega(1)} = (\log n)^{\omega(1)}$, and $2^r\vert\Sigma\vert^q = n^{O(1)}$. Putting all of this together, we can apply Theorem \ref{thm:mainresultformal} to show that $\gamma$-GapSVP$_p$ on lattices of dimension $n^{O(1)}$ is NP hard for all constants $p > 2$, where $\gamma = n^{\Theta(1)}$. Rescaling the dimension of the lattice yields the corollary.
\end{proof}

\section{Acknowledgments}
This research was supported in part from a Simons Investigator Award, DARPA SIEVE award, NTT Research, NSF grant 2333935, BSF grant 2022370, a Xerox Faculty Research Award, a Google Faculty Research Award, an Okawa Foundation Research Grant, and the Symantec Chair of Computer Science. This material is based upon work supported by the Defense Advanced Research Projects Agency through Award HR00112020024.

\newpage
\bibliographystyle{alpha}
\bibliography{refs.bib, crypto.bib}
\newpage

\appendix

\section{PCP Regularization}\label{app:regular}

For the convenience of the reader, in this section we prove the PCP Regularization theorem due to Hirahara and Moshkovitz \cite{hirahara2023regularization}, which we made use of to derive our main result:

\paragraph{Theorem \ref{thm:regtechnique} \textnormal{(Restated)}.}
    Assume that we have a PCP verifier $\mathcal{V}$ which uses $r$ random bits to make $q$ queries to a proof over alphabet $\Sigma$ and has soundness $s$. Assume further that $s \leq \min(1/(3q), 1/\vert\Sigma\vert^c)$, where $c$ is a constant with $0 < c \leq 1$. Then we can construct a new PCP verifier $\mathcal{V}'$ from $\mathcal{V}$ deterministically in $\vert \mathcal{V}\vert^{O(1)}$ time, such that the following holds. $\mathcal{V}'$ uses $r + O(\log(1/s))$ random bits to make $q^{O(1)}$ queries to a proof over the same alphabet $\Sigma$ and has soundness $s^{\Theta(1)}$. Additionally, every proof symbol is read for exactly $d = (q/s)^{\Theta(1)}$ choices of randomness.
\vspace{10px}

As in the body of our paper, we will adopt a CSP perspective (see Section \ref{sec:cspperspective}), but the proof below corresponds exactly to the proof of~\cite{hirahara2023regularization}. At a high level the plan is to construct $\mathcal{V}'$ by duplicating constraints and variables in a careful manner, and then connect constraints to the variable duplicates in accordance with an \emph{explicit disperser,} which is defined as follows.

\begin{definition}[$(\delta, \beta)$-Disperser]
    A $(\delta, \beta)$-disperser is a bi-regular bipartite graph $G = (U, V, E)$ such that for all $V' \subset V$ of size at most $\beta\vert V\vert$, at most $\delta\vert U\vert$ of the vertices $u \in U$ have all of their neighbors contained within $V'$.
\end{definition}

Hirahara and Moshkovitz's starting point for constructing an explicit disperser is the following theorem, which appears in a paper by Moshkovitz and Raz \cite{moshkovitz2008two}:

\begin{lemma}[\cite{moshkovitz2008two}] \label{lem:expander}
    There is a constant $\kappa < 1$ and a function $T(\Delta) = \Theta(\Delta)$ such that, given two natural numbers $n$ and $\Delta$, one can find in time poly$(n\Delta)$ an undirected multigraph $G = (V, E)$ where $\vert V\vert = n$, every vertex is of the same degree $T(\Delta)$, and the adjacency matrix for $G$ has its second largest eigenvalue being of magnitude at most $(T(\Delta))^\kappa$.
\end{lemma}

It's well known \cite{vadhan2012pseudorandomness, hirahara2023regularization} that we can construct explicit dispersers by taking the set of all walks of a certain length in graphs where the second largest eigenvalue is of small magnitude. Using the graph from Lemma \ref{lem:expander} in such a construction, we can additionally ensure that the degree for both partitions of the expander is fixed independently of the partition sizes.

\begin{corollary}[Simplified version of corollary in \cite{hirahara2023regularization}] \label{cor:disperser}
    There exists a function $f(w, \beta) = w \cdot (1/\beta)^{O(w)}$ such that the following holds.
    For all parameters $w \geq 1$, $\beta \in (0, 1)$, and $A \geq (1/\beta)^{q + 1}$ there exists an explicit, poly$(A)$ time construction of a $(3\beta^w, \beta)$-disperser $G = ([A], [B], E)$, where $B = \frac{wA}{f(w, \beta)}$. Every vertex in the partition indexed by $[A]$ has degree $w$ and every vertex in the partition indexed by $[B]$ has degree $f(w, \beta)$.
\end{corollary}

\begin{remark}
    The corollary holds for $w = 1$ because in this case a matching between the two partitions constitutes a $(3\beta, \beta)$-disperser. We also note that Hirahara and Moshkovitz use a slightly stronger version of the corollary, which guarantees an efficient construction of an $(e\beta^w, \beta)$-disperser, where e is Euler's number. We use the constant $3$ instead of $e$ for ease of notation.
\end{remark}

Now we demonstrate how to prove the PCP Regularization theorem by using an explicit disperser.

\begin{proof}[Proof of Theorem \ref{thm:regtechnique}]
The reduction from a PCP verifier $\mathcal{V}$ to a regularized PCP verifier $\mathcal{V}'$ goes as follows:
\begin{enumerate}
    \item Duplicate each constraint in $\mathcal{V}$ exactly $1/s$ times, which means that each variable appears in at least $1/s$ constraints. This does not affect completeness or soundness. The number of constraints increases from $2^r$ to $2^r/s$.
    \item Do the following for each variable $x$ in the PCP verifier
    \begin{enumerate}
        \item Let $d(x)$ be the number of constraints that depend on $x$, and invoke Corollary \ref{cor:disperser} where $w~=~\lceil 6q/c\rceil$, $\beta = s^{1/(2q)}$, and $A = d(x)$. (Recall that $q$ is the number of queries made by the PCP verifier $\mathcal{V}$ and $0 < c \leq 1$ is a constant such that $s \leq 1/\vert\Sigma\vert^c$.) This gives us an explicit disperser $G_x~=~([A], [B], E)$, where the degree of each vertex in the $[B]$ partition is exactly $f(w, \beta)$.
        \item Create $B$ duplicates $x_1, \ldots x_B$ of the variable $x$.
        \item Arbitrarily order the $A = d(x)$ constraints that depend on the original (un-duplicated) variable~$x$.
        \item For all $i \in [d(x)]$, modify the $i$th constraint that depends on $x$ to instead depend on the $\lceil 6q/c\rceil$ copies of $x$ that correspond with the $\lceil 6q/c\rceil$ neighbors of the $i$th vertex in partition $[A]$ in the graph $G_x$. The modified constraint now performs its original test, in addition to ensuring that all incident copies of variable $x$ are equal to each other.
    \end{enumerate}
\end{enumerate}

Observe that the construction runs in time $\vert\mathcal{V}\vert^{O(1)}$, and the new PCP verifier $\mathcal{V}'$ immediately satisfies most of the conditions in Theorem \ref{thm:regtechnique}. In particular, $\mathcal{V}'$ uses $\log(2^r/s) = r + O(\log(1/s))$ bits of randomness, makes $q \cdot \lceil 6q/c\rceil = O(q^2)$ queries to a proof over the same alphabet as before, and every proof symbol is read for exactly $f(w, \beta) = q \cdot (1/s^{1/(2q)})^{O(q)} = (q/s)^{O(1)}$ choices of randomness. Completeness is also immediate: If $\mathcal{V}$ has a satisfying assignment, then we can simply assign each copy of each variable in accordance with the assignment for $\mathcal{V}$, and every constraint of $\mathcal{V}'$ is satisfied.

All that remains is to demonstrate soundness.

\begin{claim}
    Suppose that the original PCP verifier $\mathcal{V}$ was not satisfiable, i.e. every assignment to the variables satisfies at most an $s$ fraction of the constraints. Then every assignment to the variables in $\mathcal{V}'$ satisfies at most a $2\sqrt{s}$ fraction of the constraints.
\end{claim}

\begin{proof}
    Suppose for contradiction that there exists an assignment to the variables of $\mathcal{V}'$ that satisfies more than a $2\sqrt{s}$ fraction of the constraints. For each original variable $x$, consider the corresponding disperser $G_x = ([A], [B], E)$ and examine the set of variable copies $x_1, \ldots x_B$. Let $\Sigma_{x, \text{high}}$ be the $(1/s)^{1/(2q)}$ most popular assignments from $\Sigma$ given to the copies of variable $x$, and let $\Sigma_{x, \text{low}} = \Sigma \backslash \Sigma_{x, \text{high}}$. We say that a constraint of $\mathcal{V}'$ is bad if it is satisfied while simultaneously being incident to a copy of at least one variable $x$ that was assigned to a symbol in $\Sigma_{x, \text{low}}$.

    We first argue that, regardless of whether the original PCP verifier $\mathcal{V}$ was satisfiable, at most an $s$ fraction of the constraints in $\mathcal{V}'$ are bad. First consider any variable $x$ and any $\alpha \in \Sigma_{x, \text{low}}$. Notice that by definition of $\Sigma_{x, \text{low}}$, at most an $s^{1/(2q)}$ fraction of the copies of $x$ can be assigned to $\alpha$. This value of $s^{1/(2q)}$ is the parameter $\beta$ we chose when constructing our explicit dispersers. So by the definition of disperser, we know that at most a $3\beta^w$ fraction of $\mathcal X_{x,\alpha} = \{$the constraints in $\mathcal{V}'$ that are incident to a copy of $x \}$ will depend only on copies assigned to $\alpha$. These are the only constraints incident to at least one copy assigned to $\alpha$ that could possibly be satisfied, because our constraints check equality between all of their incident copies. Phrased more explicitly, at most a
    \[3\beta^w \leq 3(s^{1/(2q)})^{\lceil 6 q/c\rceil} \leq 3s^{3/c}\]
    fraction of the constraints in $\mathcal{V}'$ that are incident to a copy of $x$ will simultaneously (i) be satisfied and (ii) have the copy assigned to $\alpha$. Union bounding over the at most $\vert\Sigma\vert$ choices for $\alpha \in \Sigma_{x, \text{low}}$, and using the assumptions that $s \leq \min(1/(3q), 1/\vert\Sigma\vert^c)$ and $0 < c \leq 1$, at most a
    \[3s^{3/c} \cdot \vert\Sigma\vert \leq 3s^{2/c} \leq 3s^2 \leq s/q\]
    fraction of the constraints in $\mathcal{V}'$ incident to a copy of $x$ will simultaneously (i) be satisfied and (ii) have the copy assigned to any value in $\Sigma_{x, \text{low}}$. Because each constraint is incident to copies corresponding to $q$ different original variables, and each set of copies could cause the constraint to be bad, a final union bound over the arity $q$ gives the desired upper bound, that at most an $s$ fraction of the constraints in $\mathcal{V}'$ are bad.

    Because we assumed that we have an assignment which satisfies more than a $2\sqrt{s}$ fraction of the constraints in $\mathcal{V}'$, there must be more than a $2\sqrt{s} - s \geq \sqrt{s}$ fraction of constraints in $\mathcal{V}'$ that are satisfied and not bad. In other words, more than a $\sqrt{s}$ fraction of the constraints in $\mathcal{V}'$ are satisfied while simultaneously being incident only to copies of variables $x$ that are assigned to symbols in $\Sigma_{x, \text{high}}$. By definition, for all variables $x$ we have that $\vert\Sigma_{x, \text{high}}\vert \leq (1/s)^{1/(2q)}$. If we randomly pick a value in $\Sigma_{x, \text{high}}$ for each variable $x$ and take this as an assignment for the original PCP verifier $\mathcal{V}$, we expect to satisfy strictly more than a $\sqrt{s} \cdot \left(s^{1/(2q)}\right)^{q} = \sqrt{s} \cdot s^{1/2} = s$ fraction of the constraints in $\mathcal{V}$. Taking the best choice of randomness gives a deterministic assignment which contradicts the soundness of $\mathcal{V}$.
\end{proof}

With both completeness and soundness proven, the theorem follows.
\end{proof}

\section{Proof of Corollary \ref{cor:Holder}}
\label{app:Holder}

Recall Hölder's Inequality.

\paragraph{Theorem \ref{thm:Holder} \textnormal{(Restated)}.}
    \emph{Let $\mat u, \mat v \in \mathbb{Z}^n$, and let $p, q \geq 1$ be parameters satisfying $1/p + 1/q = 1$. Then}
    \[\|\mat u \odot \mat v\|_1 \leq \|\mat u\|_p \|\mat v \|_q,\]
    \emph{where $\odot$ denotes the component-wise product.}\\

We now prove Corollary \ref{cor:Holder}

\paragraph{Corollary \ref{cor:Holder} \textnormal{(Restated)}.}
    \emph{For all $\mat w \in \mathbb{Z}^n$ and $p \geq 2$ (including $p = \infty$),}
    \[\|\mat w\|_p \geq n^{1/p - 1/2}\|\mat w\|_2.\]

\begin{proof}
    We start with the special cases $p = 2$ and $p = \infty$. Notice that when $p = 2$, we have $n^{1/p - 1/2} = 1$, so the inequality holds immediately. Now consider $p = \infty$, and assume for contradiction that $\|\mat w\|_\infty < n^{-1/2}\|\mat w\|_2$. Then for all $i \in [n]$, we have $\vert\mat w_i\vert < n^{-1/2}\|\mat w\|_2$, and the $\ell_2$ norm is upper bounded as
    \begin{align*}
        \|\mat w\|_2 < & \left(\sum_{i = 1}^n n^{-1}(\|\mat w\|_2)^2\right)^{1/2} \\
        = & \left((\|\mat w\|_2)^2\right)^{1/2} \\
        = & \|\mat w\|_2. \\
    \end{align*}
    Clearly $\|\mat w\|_2 < \|\mat w\|_2$ is a contradiction.

    Now consider the case that $p > 2$ is finite. Set $\mat v \in \mathbb{Z}^n$ to be $\mat w \odot \mat w$, i.e. we square each component. Define $\mat u = \mat 1^n$, and set $a = p/(p-2) > 1$ and $b = p/2 > 1$. Because $1/a + 1/b = (p-2)/p + 2/p = 1$, an application of Theorem \ref{thm:Holder} gives
    \[\|\mat u \odot \mat v\|_1 \leq \|\mat u\|_a \|\mat v\|_b.\]
    By the definition of $p$ norm we can write this as
    \[\sum_{i = 1}^n \vert\mat u_i \mat v_i\vert \leq \left(\sum_{i = 1}^n \vert \mat u_i^a\vert\right)^{1/a} \left(\sum_{i = 1}^n \vert \mat v_i^b\vert\right)^{1/b}.\]
    Now because $\mat v$ is the component-wise square of $\mat w$, and $\mat u$ is the all-ones vector, we have
    \begin{align*}
        \sum_{i = 1}^n \vert\mat u_i \mat v_i\vert \leq & \left(\sum_{i = 1}^n \vert \mat u_i^a\vert\right)^{1/a} \left(\sum_{i = 1}^n \vert \mat v_i^b\vert\right)^{1/b} \\
        \sum_{i = 1}^n \vert\mat w_i^2\vert \leq & \left(\sum_{i = 1}^n 1\right)^{1/a} \left(\sum_{i = 1}^n \vert \mat w_i^{2b}\vert\right)^{1/b} \\
        \sum_{i = 1}^n \vert\mat w_i^2\vert \leq & n^{(p - 2)/p} \left(\sum_{i = 1}^n \vert \mat w_i^{p}\vert\right)^{2/p} \\
        (\|\mat w\|_2)^2 \leq & n^{(p-2)/p} (\|\mat w\|_p)^2.
    \end{align*}
    Re-arranging and taking square roots yields
    \begin{align*}
        n^{(p-2)/p} (\|\mat w\|_p)^2 \geq & (\|\mat w\|_2)^2 \\
        (\|\mat w\|_p)^2 \geq & n^{(2-p)/p}(\|\mat w\|_2)^2 \\
        \|\mat w\|_p \geq & n^{1/p - 1/2}\|\mat w\|_2.
    \end{align*}
\end{proof}

\end{document}